\pdfoutput=1
\documentclass[letterpaper,11pt]{article}
\usepackage[margin=1in]{geometry}
\usepackage{amsmath, amsfonts, amssymb, amsthm}
\usepackage{thm-restate}
\usepackage{bbm,url}
\usepackage{graphicx}
\usepackage{xcolor}
\usepackage{colortbl}
\usepackage{subcaption}
\usepackage{hhline}
\usepackage{pifont}
\usepackage{multirow}
\usepackage{multicol}
\usepackage{enumitem}
\usepackage{comment}
\usepackage{nicefrac,bm}
\usepackage{algorithm}
\usepackage[noend]{algpseudocode}
\usepackage{thmtools}
\usepackage{thm-restate}

\newcommand\NoDo{\renewcommand\algorithmicdo{}}
\newcommand\ReDo{\renewcommand\algorithmicdo{\textbf{do}}}

\allowdisplaybreaks

\newtheorem{theorem}{Theorem}
\newtheorem*{theorem*}{Theorem}
\newtheorem{corollary}[theorem]{Corollary}

\newtheorem{lemma}{Lemma}

\newtheorem{remark}{Remark}
\newtheorem{observation}{Observation}
\newtheorem{example}[theorem]{Example}
\newtheorem{invariant}{Invariant}

\theoremstyle{definition}

\newcommand{\hide}[1]{}

\newcommand{\s}[1]{\mathsf{#1}}

\newcommand{\mpb}{\mathsf{MPB}}

\newcommand{\x}{\mathbf{x}}
\newcommand{\y}{\mathbf{y}}
\newcommand{\p}{\mathbf{p}}

\newcommand{\R}{\mathbb R}
\newcommand{\poly}[1]{\mathsf{poly}(#1)}

\newcommand{\RR}{$\mathsf{RoundRobin}$ }

\title{New Algorithms for the Fair and Efficient Allocation of Indivisible Chores\thanks{A preliminary version of this paper appeared at IJCAI 2023 \cite{GMQ23chores}. Work supported by NSF Grant CCF-1942321.}} 

\author{Jugal Garg\footnote{University of Illinois at Urbana-Champaign, USA} \\
\texttt{\small jugal@illinois.edu} 
\and
Aniket Murhekar\footnote{University of Illinois at Urbana-Champaign, USA}\\
\texttt{\small aniket2@illinois.edu}
\and
John Qin\footnote{University of Illinois at Urbana-Champaign, USA}\\
\texttt{\small johnqin2@illinois.edu}
}

\date{}
\begin{document}
\renewcommand{\arraystretch}{1.2}
\maketitle

\begin{abstract}
We study the problem of fairly and efficiently allocating indivisible chores among agents with additive disutility functions. We consider the widely-used envy-based fairness properties of EF1 and EFX, in conjunction with the efficiency property of fractional Pareto-optimality (fPO). Existence (and computation) of an allocation that is simultaneously EF1/EFX and fPO are challenging open problems, and we make progress on both of them. We show the existence of an allocation that is
\begin{itemize}
    \item EF1+fPO, when there are three agents,
    \item EF1+fPO, when there are at most two disutility functions,
    \item EFX+fPO, for three agents with bivalued disutilities.
\end{itemize}
These results are constructive, based on strongly polynomial-time algorithms. We also investigate non-existence and show that an allocation that is EFX+fPO need not exist, even 
for two agents.
\end{abstract}

\section{Introduction}
Discrete fair division has recently received significant attention due to its applications in a wide variety of multi-agent settings; see recent surveys \cite{aziz2020survey,walsh2020survey,aziz2022survey}. Given a set of indivisible items and a set of $n$ agents with diverse preferences, the goal is to find an allocation that is \emph{fair} (i.e., acceptable by all agents) and \emph{efficient} (i.e., non-wasteful). We assume that agents have additive valuations. The standard economic efficiency notion is Pareto-optimality (PO) and its strengthening fractional Pareto-optimality (fPO). Fairness notions based on \textit{envy} \cite{foleyEF} are most popular, where an allocation is said to be \emph{envy-free} (EF) if every agent weakly prefers her bundle to any other agent's bundle. Since EF allocations need not exist (e.g., dividing one item among two agents), its relaxations \emph{envy-free up to any item} (EFX) \cite{caragiannis2019mnwacm} and \emph{envy-free up to one item} (EF1) \cite{lipton,budish2011approxCEEI} are most widely used, where EF $\Rightarrow$ EFX $\Rightarrow$ EF1.

\begin{table*}[t]
\centering
\begin{tabular}{|c||c|c||c|c|}\hline
 \multirow{2}{*}{\bf Instance type} & \multicolumn{2}{c||}{\bf EF1+fPO} & \multicolumn{2}{c|}{\bf EFX+fPO}  \\\hhline{|~||-|-|-|-|}
 & Goods & Chores & Goods & Chores \\\hhline{|=#=|=#=|=|}
 General additive & \checkmark~BKV\cite{Barman18FFEA}, GM\cite{GargM21}$^\dagger$ & ? & \ding{55}~GM\cite{garg2021sagt} & \cellcolor{blue!25} \ding{55} (Thm. \ref{thm:efx-fpo-nonexistence}) \\\hhline{|-|-|-|-|-|}
 $n=3$ agents & \checkmark GM\cite{GargM21} & \cellcolor{blue!25}\checkmark (Thm.~\ref{thm:3agents}) & \cellcolor{blue!25} \ding{55} (Thm. \ref{thm:efx-fpo-nonexistence}) &
 \cellcolor{blue!25} \ding{55} (Thm. \ref{thm:efx-fpo-nonexistence})\\\hhline{|-|-|-|-|-|}
 Two-type & \cellcolor{blue!25}\checkmark (Rem.~\ref{rem:2types}) & \cellcolor{blue!25}\checkmark (Thm.~\ref{thm:2types}) & \cellcolor{blue!25} \ding{55} (Thm. \ref{thm:efx-fpo-nonexistence}) & \cellcolor{blue!25} \ding{55} (Thm. \ref{thm:efx-fpo-nonexistence})\\\hhline{|-|-|-|-|-|}
 Bivalued & \checkmark GM\cite{garg2021sagt} & \checkmark GM\cite{Garg_Murhekar_Qin_2022}, EPS\cite{ebadian2021bivaluedchores} & \checkmark GM\cite{garg2021sagt} & \cellcolor{blue!25}\checkmark $n=3$ (Thm. \ref{thm:bivalued3agents}) \\\hhline{|-|-|-|-|-|}
 2-ary & \checkmark GM\cite{GargM21} & \cellcolor{blue!25}\checkmark $k_i \ge m$ (Lem.~\ref{lem:BC_EF1_PO_2ary})  & \cellcolor{blue!25}\ding{55} (Thm. \ref{thm:efx-fpo-nonexistence}) & \cellcolor{blue!25}\ding{55} (Thm. \ref{thm:efx-fpo-nonexistence})\\\hline
 \end{tabular}
\caption{State-of-the-art for EF1/EFX+fPO allocation of indivisible items. \checkmark denotes existence/polynomial-time algorithm, \ding{55} denotes non-existence, ? denotes (non-)existence is unknown, $\dagger$ denotes no polynomial-time algorithm is known. Colored cells highlight our results.}\label{tab:results}
\end{table*}

Achieving both fairness and efficiency is utmost desirable because just an efficient allocation can be highly unfair, and similarly, just a fair allocation can be highly inefficient\footnote{Consider an example with 2 agents and 2 items where agent 1 prefers item 1 to 2 while agent 2 prefers item 2 to 1. Giving all items to one agent is an efficient but unfair allocation. Similarly, giving item 2 to agent 1 and item 1 to agent 2 is a fair, according to EFX/EF1, but inefficient allocation.}. However, attaining both may not even be possible because these are often conflicting requirements, which put hard constraints on the set of all feasible allocations. Moreover, the landscape of known existence and tractability results varies depending on the nature of the items. The items to be divided can be either goods (which provide utility) or chores (which provide disutility). 

For the case of goods, a series of works provided many remarkable results showing the existence of EF1+(f)PO allocations. There are two broad approaches. The first uses the concept of \textit{Nash welfare}, which is the geometric mean of agent utilities. Caragiannis et al.~\cite{caragiannis2019mnwacm} showed that the allocation with the maximum Nash welfare (MNW) is both EF1 and PO. However, computing MNW allocation is APX-hard~\cite{lee2015nsw-apx,garg2017nswhardness}, thus rendering this approach ineffective for computing an EF1+PO allocation in general. A notable special case is when agents have binary valuations; the MNW allocation can be computed in polynomial time for this case~\cite{darmann2014binary}. The efficiency notion of PO does not seem appropriate for fast computation because even checking if an allocation is PO is a coNP-hard problem~\cite{keijzer}. Therefore, for fast computation, we need to work with the stronger notion of fPO, which admits polynomial time verification.

The second approach achieves fairness while maintaining efficiency through a \emph{competitive equilibrium} \cite{Barman18FFEA,barman2019prop1po,GargM21,garg2021sagt} of a Fisher market. In a Fisher market, agents are endowed with some monetary budget, which they use to `buy' goods to maximize their utility. A competitive equilibrium is an allocation along with the prices of the goods in which each agent only owns goods that give them maximum `bang-for-buck', i.e., goods with the highest utility-to-price ratio. The latter property ensures that the resulting allocation is fPO. The idea is then to endow agents with fictitious budgets and maintain an allocation that is the outcome of a Fisher market while performing changes to the allocation, prices, and budgets to achieve fairness. Barman, Krishnamurthy, and Vaish~\cite{Barman18FFEA} used this approach to show that an EF1+fPO allocation exists and that an EF1+PO allocation can be computed in pseudo-polynomial time. Later, Garg and Murhekar~\cite{GargM21} showed that an EF1+fPO allocation can be computed in pseudo-polynomial time in general and in polynomial time for constantly many agents or when agents have $k$-ary valuations with constant $k$. For bivalued instances of goods, an EFX+fPO allocation was shown to be polynomial time computable \cite{garg2021sagt}. Designing a polynomial-time algorithm for computing an EF1+fPO allocation of goods remains a challenging open problem.

In contrast, the case of chores turns out to be much more difficult to work with, resulting in relatively slow progress despite significant efforts by many researchers. Neither of the above-mentioned approaches seems to be directly applicable in the setting of chores. 
Indeed, given the absence of a global welfare function like Nash welfare, even the existence of EF1+PO allocations for chores is open. Using competitive equilibria for chores remains a promising approach, which also guarantees fPO. In a Fisher market for chores, agents have a monetary expectation, i.e., a salary, which they aim to achieve by performing chores that have associated \textit{payments} instead of prices. The algorithms for computing an EF1+(f)PO allocation of goods (\cite{Barman18FFEA,GargM21}) use the CE framework and show termination via some potential function. The main difficulty in translating algorithms for goods to the chores setting seems to be that the price-rises and item transfers only increase the potential in the case of the goods. For chores, however, price (payment) changes and transfers do not push these potential functions in one direction as they do for goods, making it difficult to show termination.

Consequently, the existence of EF1 + (f)PO allocation for chores remains open except for the case of two agents~\cite{aziz2019chores}, bivalued instances~\cite{ebadian2021bivaluedchores,Garg_Murhekar_Qin_2022}, and two types of chores \cite{aziz2022twotypes}. 
The problem becomes significantly difficult when there are $n> 2$ agents. Table~\ref{tab:results} provides a summary of existing results that are relevant to our work. In this paper, we focus on the chores setting and make progress on the above-mentioned problems. Our first set of results show that an allocation that is 

\begin{itemize}[leftmargin=*]
\item EF1+fPO exists when there are three agents. Our algorithm uses the competitive equilibrium framework (CE) to maintain an fPO allocation while using the payments associated with chores to guide their transfer to reduce the envy between the agents. Our novel approach starts with one agent having the highest envy and then makes careful chore transfers unilaterally away from this agent while maintaining that the other two agents have bounded envy. 

\item EF1+fPO exists when there are two types of agents, where agents of the same type have the same preferences. This subsumes the result of Aziz et al.~\cite{aziz2019chores} computing EF1+fPO for two agents. We develop a novel approach combining the CE allocation of an appropriately constructed market with a round-robin procedure. Combining CE-based frameworks with envy-resolving algorithms may be an important tool in settling the problem in its full generality. Our approach also gives a similar result for the case of goods. Note that two types of agents strictly generalize the well-studied setting of identical agents~\cite{barman2018binarynsw,plaut2018efx}. Mahara~\cite{mahara21efx} showed that EFX, without fPO, exists for two types of agents in the case of goods through involved case analysis.

\item EFX+fPO exists when there are three agents with bivalued preferences, where each disutility value is one of two values. This improves the result of Zhou and Wu~\cite{zhou22efx} which shows that EFX exists in this case. Similar to~\cite{zhou22efx}, our algorithm is quite involved based on a case-by-case analysis. We first derive a simple algorithm for computing an EF1+fPO for bivalued preferences, which was recently shown to exist by~\cite{Garg_Murhekar_Qin_2022,ebadian2021bivaluedchores}. An interesting aspect of the simple algorithm is that it outputs an allocation that gives each agent a \emph{balanced} number of chores. We start from such a balanced EF1+fPO outcome and improve the guarantee from EF1 to EFX while maintaining the fPO property. Additionally, we show that EF1+PO exists for a class of 2-ary preferences, where each disutility value of an agent is one of two values, but these two values can be different for different agents. This class is not subsumed by bivalued instances.
\end{itemize}

All our existence results are accompanied by polynomial-time algorithms. Next, we investigate the non-existence of fair and efficient allocations and show that 
\begin{itemize}[leftmargin=*]
\item EFX+fPO need not exist when there are two agents with 2-ary disutility functions. Naturally, this also implies that an EFX+fPO allocation need not exist for two-type instances as well.
\end{itemize}

\subsection{Further Related Work}
The problems in their full generality have remained challenging open questions. Most progress has been made for special cases, providing a better picture of (non-)existence and (in)tractable classes of discrete fair division.

\paragraph{EF1.} Barman, Krishnamurthy, and Vaish~\cite{Barman18FFEA} showed that an EF1+PO allocation of goods is computable in pseudo-polynomial time. Garg and Murhekar~\cite{GargM21} showed algorithms for computing an EF1+fPO allocation in (i) pseudo-polynomial time for general additive goods, (ii) poly-time for $k$-ary instances with constant $k$, (iii) poly-time for constantly many agents. Barman, Krishnamurthy, and Vaish \cite{barman2018binarynsw} show the polynomial time computation of EF1+fPO allocation for binary instances of goods. For chores, it is not known if EF1+PO allocations exist. Recently, the existence and polynomial time computation of an EF1+(f)PO allocation was shown for two agents by Aziz et al. \cite{aziz2019chores} for bivalued chores by Garg, Murhekar, and Qin~\cite{Garg_Murhekar_Qin_2022}, and Ebadian, Peters, and Shah~\cite{ebadian2021bivaluedchores}, and for two types of chores by \cite{aziz2022twotypes}.

\paragraph{EFX.} Plaut and Roughgarden~\cite{plaut2018efx} showed that an EFX+PO allocation exists for the case of identical goods using the leximin mechanism. Garg and Murhekar~\cite{garg2021sagt} showed an EFX+fPO allocation can be computed in poly-time for bivalued goods, and non-existence for 3-valued instances. For both goods and chores, the existence of EFX allocations is a challenging open problem. Existence is known in the goods case for three agents \cite{chaudhury2020efx,AkramiCGMM22} and two types of agents \cite{mahara21efx}, and for two types of goods \cite{gorantla2022efx}. In the chores case, EFX allocations are known to exist for three bivalued agents \cite{zhou22efx}, or when there are two types of chores~\cite{aziz2022twotypes}.

\paragraph{Organization of the paper.} Section~\ref{sec:prelim} introduces relevant notation and definitions.  Sections~\ref{sec:3agents} and \ref{sec:2types} discuss our algorithms for computing an EF1+fPO allocation for three agents and two-type instances, respectively. Section~\ref{sec:bivalued3agents} presents an overview of our algorithm computing an EFX+fPO allocation for three bivalued agents, which is described and analyzed in full detail in Appendix~\ref{app:3bivaluedagents}. Appendix \ref{app:examples} contains illustrative examples.

\section{Notation and Preliminaries}\label{sec:prelim}
In a \textit{chore allocation} instance $(N, M, D)$, we are given a set $N = [n]$ of $n$ agents, a set $M = [m]$ of $m$ indivisible chores, and a set $D = \{d_i\}_{i \in [n]}$, where $d_i : 2^M \rightarrow \R_{\ge 0}$ is agent $i$'s \textit{disutility} or \textit{cost} function over the chores. We assume that $d_i(\varnothing) = 0$. We let $d_i(j)$ denote the disutility agent $i$ incurs from chore $j$ and assume disutility functions are additive, so that for every $i \in N$ and $S \subseteq M$, $d_i(S) = \sum_{j \in S} d_i(j)$. In addition to general chore allocation instances, we especially consider the following classes. We call a chore allocation instance a:

\begin{itemize}
\item \textit{Three agent instance} if there are $n=3$ agents.
\item \textit{Two-type instance} if there exist disutility functions $d_1$ and $d_2$ so that for all $i \in N$, $d_i \in \{d_1, d_2\}$. That is, every agent has one of two unique disutility functions.
\item \textit{Bivalued instance} if there exist $a, b \in \R_+$ so that for all $i \in N$ and $j \in M$ we have $d_i(j) \in \{a, b\}$. Here, rather than two disutility functions, we have two chore costs. 
\item \textit{2-ary instance} if for each $i \in N$ there exist $a_i, b_i \in \R_+$ so that for all $j \in M$ we have $d_i(j) \in \{a_i, b_i\}$. Clearly, 2-ary instances strictly generalize bivalued instances. However, neither is comparable with two-type instances.
\end{itemize}

An \textit{allocation} $\x = (\x_1, \x_2, \ldots, \x_n)$ is an $n$-partition of the chores where agent $i$ receives bundle $\x_i \subseteq M$ and incurs disutility $d_i(\x_i)$. A \textit{fractional allocation} $\x \in [0, 1]^{n \times m}$, where $\x_{ij} \in [0, 1]$ denotes the fraction of chore $j$ given to agent $i$, allows for the chores to be divided. Here $d_i(\x_i) = \sum_{j \in M} d_i(j) \cdot x_{ij}$. We will assume that allocations are integral unless explicitly stated otherwise.

We next define our fairness notions. An allocation $\x$ is said to be:
\begin{itemize}
\item \textit{Envy-free} if for all $i, h \in N$, $d_i(\x_i) \leq d_i(\x_h)$, i.e., every agent weakly prefers her own bundle to others' bundles. 
\item \textit{Envy-free up to any chore} (EFX) if for all $i,h\in N$, $d_i(\x_i \setminus j) \leq d_i(\x_h)$ for any $j\in \x_i$, i.e., every agent weakly prefers her own bundle to any other agent's bundle after removing her own easiest (least disutility) chore. 
\item \textit{Envy-free up to one chore} (EF1) if for all $i,h\in N$, $d_i(\x_i \setminus j) \leq d_i(\x_h)$ for some $j\in \x_i$, i.e., every agent weakly prefers her own bundle to any other agent's bundle after removing her own hardest (highest disutility) chore. We use $d_{i_{-1}}(S)$ to denote $\min_{j \in S} d_i(S \setminus j)$ for $S\subseteq M$. Thus $\x$ is EF1 if $\forall i, h \in N$, $d_{i_{-1}}(\x_i) \leq d_i(\x_h)$.
\end{itemize}

We now define the efficiency notions of Pareto optimality (PO) and fractional Pareto optimality (fPO). An allocation $\y$ dominates an allocation $\x$ if for all $i \in N$, $d_i(\y_i) \leq d_i(\x_i)$, and there exists $h \in N$ such that $d_h(\y_h) < d_h(\x_h)$. An allocation is then PO if it is not dominated by any other allocation. Similarly, an allocation is fPO if it is not dominated by any fractional allocation. Note that an fPO allocation is necessarily PO, but not vice-versa.

Two important concepts we use in our algorithms are \textit{competitive equilibrium} and \textit{Fisher markets}. In the Fisher market setting we attach payments $\p = (p_1, \ldots, p_m)$ to the chores. Agents perform chores in exchange for payment, with each agent $i$ aiming to earn her \textit{minimum payment} $e_i \geq 0$. Given a (fractional) allocation $\x$ and a set of payments $\p$, the \textit{earning} of an agent $i$ under $(\x, \p)$ is given by $\p(\x_i) = \sum_{j \in M} p_j \cdot x_{ij}$. For each agent $i$, we define the \textit{pain-per-buck} ratio $\alpha_{ij}$ of chore $j$ as $\alpha_{ij} = d_i(j) / p_j$ and the \textit{minimum-pain-per-buck} (MPB) ratio $\alpha_i = \min_{j \in M} \alpha_{ij}$. Further, we let $\mpb_i = \{j \in M \mid d_i(j) / p_j = \alpha_i\}$ denote the set of chores which are MPB for agent $i$ for payments $\p$. 

We say that $(\x, \p)$ is a \textit{competitive equilibrium} if (i) for all $j \in M$, $\sum_{i \in N} x_{ij} = 1$, i.e., all chores are completely allocated, (ii) for all $i \in N$, $\p(\x_i) = e_i$, i.e., each agent receives her minimum payment, and (iii) for all $i \in N$, $\x_i \subseteq \mpb_i$, i.e., agents receive only chores which are MPB for them. Competitive equilibria are known to guarantee economic efficiency via the First Welfare Theorem \cite{mas1995microeconomic}, i.e., for a competitive equilibrium $(\x, \p)$, the allocation $\x$ is fPO.

Given a competitive equilibrium $(\x, \p)$ with integral allocation $\x$, we let $\p_{-1}(\x_i)$ denote the payment agent $i$ receives from $\x_i$ excluding her highest paying chore. That is, $\p_{-1}(\x_i) = \min_{j \in \x_i} \p(\x_i \setminus j)$. We say that $(\x, \p)$ is \textit{payment envy-free up to one chore} (pEF1) if for all $i, h \in N$ we have $\p_{-1}(\x_i) \leq \p(\x_h)$. We say that agent $i$ \textit{pEF1-envies} $h$ if $\p_{-1}(\x_i) > \p(\x_h)$. The following lemma shows that the pEF1 condition is, in fact, a strengthening of EF1. 

\begin{lemma}\label{lem:pEF1impliesEF1}
Let $(\x,\p)$ be an integral competitive equilibrium. If $(\x,\p)$ is pEF1, then $\x$ is EF1+fPO.
\end{lemma}
\begin{proof}
As $(\x,\p)$ is pEF1, for all agents $i,h \in N$ we have $\p_{-1}(\x_i)\leq \p(\x_h)$. Since an agent $i$ has only MPB chores, i.e., $\x_i \subseteq \mpb_i$,  $d_{i_{-1}}(\x_i) = \alpha_i \p_{-1}(\x_i)$ and $d_i(\x_h) \geq \alpha_i \p(\x_h)$, where $\alpha_i$ is the MPB ratio of agent $i$. This gives us $d_{i_{-1}}(\x_i) = \alpha_i \p_{-1}(\x_i) \leq \alpha_i \p(\x_h) \leq d_i(\x_h)$, showing that $\x$ is EF1. Additionally, the First Welfare Theorem \cite{mas1995microeconomic} implies that the allocation $\x$ is fPO for a competitive equilibrium $(\x, \p)$.
\end{proof}

Lemma~\ref{lem:pEF1impliesEF1} suggests that in order to compute an EF1+fPO allocation, one can search instead for a pEF1+fPO allocation, if possible. To do this, a natural approach is to start with an integral competitive equilibrium $(\x,\p)$, and perform changes to both the allocation $\x$ and associated payments $\p$ until the allocation is EF1 (or pEF1), while maintaining that $(\x,\p)$ remains an integral competitive equilibrium. The latter ensures that $\x$ is always fPO. We, therefore, perform a transfer of a chore $j \in \x_h$ from an agent $h$ to an agent $i$ only if $j\in \mpb_i$. If we wish to take a chore away from an agent $h$ to reduce her envy and yet for all agents $i$ we have $\x_h \cap \mpb_h = \emptyset$, then no agent $i$ can receive a chore of $h$. In such a case, we must adjust the payments to ensure some chore of $h$ enters the MPB set of another agent $i$. This can be done in one of two ways. The first way is to \textit{reduce the payments} of chores owned by a subset of agents other than $h$, done in a manner that does not violate the MPB condition of any agent until the agents undergoing \textit{payment drops} find the chores of $h$ attractive enough to enter their MPB set. Alternatively, one could \textit{raise the payments} of chores of $h$, along with possibly some other chores to ensure that the MPB condition of no agent is violated until an agent who did not undergo a \textit{payment rise} finds the chores of $h$ attractive enough to enter their MPB set. Once there is an agent $i$ such that there is a chore $j\in\x_h \cap \mpb_i$, we can transfer $j$ from $h$ to $i$. Our goal is to perform chore transfers facilitated via payment drops or raises to modify an initial allocation to one that is EF1 (or pEF1).

Two agents will be of particular interest to execute the above ideas. An agent $\ell$ is called a \textit{least earner} (LE) among agent set $A$ if $\ell \in \s{argmin}_{i \in A} \p(\x_i)$. An agent $b$ is called a \textit{big earner} (BE) among agent set $A$ if $b \in \s{argmax}_{i \in A} \p_{-1}(\x_i)$, i.e., among agents in $A$, $b$ earns the highest payment from her bundle of chores after every agent removes their highest paying chore. We call an agent the \textit{global} least/big earner when $A = N$. We will assume that a mentioned least earner or big earner is global unless explicitly noted otherwise. The next lemma shows the importance of the BE and LE agents.

\begin{lemma}\label{lem:BEenvyLE}
An integral competitive equilibrium $(\x,\p)$ is pEF1 if and only if a big earner $b$ does not pEF1-envy a least earner $\ell$.
\end{lemma}
\begin{proof}
It is clear that if $(\x, \p)$ is pEF1 then $b$ does not pEF1-envy $\ell$. We need only show that if $b$ does not pEF1-envy $\ell$, then $(\x, \p)$ is pEF1. For all $i, h \in N$, we have that $\p_{-1}(\x_i) \leq \p_{-1}(\x_b)$ by the definition of BE, and that $\p(\x_{\ell}) \leq \p(\x_h)$ by the definition of LE. Putting these together with $\p_{-1}(\x_b) \le \p(\x_h)$ since $b$ does not pEF1-envy $\ell$, we get that $\p_{-1}(\x_i) \le \p(\x_h)$. Thus $i$ does not pEF1-envy $h$ and $(\x, \p)$ is pEF1.
\end{proof}

\section{EF1 + fPO for Three Agents}\label{sec:3agents}

In this section, we prove the first main result of our paper. 

\begin{theorem}\label{thm:3agents}
Given a chore allocation instance $(N,M,D)$ with three agents, an EF1 + fPO allocation exists. Furthermore, it can be computed in strongly polynomial time.
\end{theorem}

We prove Theorem~\ref{thm:3agents} by showing that Algorithm~\ref{alg:3agents} computes an EF1 + fPO allocation in polynomial time. Algorithm~\ref{alg:3agents} begins by allocating the entire set of chores $M$ to an arbitrarily chosen agent $i \in N$, with payments set so that $p_j = d_i(j)$, giving $\alpha_i = 1$.\footnote{We can assume w.l.o.g. that $d_i(j) > 0$ for all $i,j$. Otherwise, if $d_i(j)=0$ for some $i,j$, we can simply allocate $j$ to $i$ and remove $j$ from further consideration. It is easy to check that doing so does not affect EF1 or fPO properties.} This gives us an initial competitive equilibrium $(\x, \p)$ where $i$ is the big earner (BE). We define agent $\ell$ to be the least earner (LE) and agent $h$ to be the \textit{middle earner} (ME), i.e., the agent who is neither the big earner nor the least earner. In the initial allocation, $\ell$ and $h$ are arbitrarily chosen after the BE $i$ is chosen.

\begin{algorithm}[!t]
\caption{Computing an EF1+fPO allocation for 3 agents}\label{alg:3agents}
\textbf{Input:} Fair division instance $(N,M,D)$ with $\lvert N \rvert = 3$ \\
\textbf{Output:} An integral allocation $\x$
\begin{algorithmic}[1]
\State $\x \gets$ give $M$ to some agent $i\in\{1,2,3\}$ 
\For{$j \in M$}
    \State $p_j \gets d_i(j)$
\EndFor
\While{$\x$ is not EF1}
    \State $b \gets \s{argmax}_{k \in N} \p_{-1}(\x_k)$ \Comment{Big earner}
    \State $\ell \gets \s{argmin}_{k \in N} \p(\x_k)$ \Comment{Least earner}
    \State $h \gets k \in N \setminus \{b, \ell\}$ \Comment{Middle earner}
    \If{$\exists j \in \x_b \cap \s{MPB}_{\ell}$} \Comment{A chore can be potentially transferred from $b$ to $\ell$}  
        \State $\x_b \gets \x_b \setminus j$
        \State $\x_{\ell} \gets \x_{\ell} \cup j$
    \ElsIf{$\exists j \in \x_h \cap \s{MPB}_{\ell}$} \Comment{A chore can be potentially transferred from $h$ to $\ell$}
        \If{$\p(\x_h \setminus j) > \p(\x_{\ell})$} \Comment{$h$ pEF1-envies $\ell$ w.r.t. any chore $j\in \x_h\cap \mpb_\ell$ }
            \State $\x_h \gets \x_h \setminus j$
            \State $\x_{\ell} \gets \x_{\ell} \cup j$ 
        \ElsIf{$\exists j' \in \x_b \cap \s{MPB}_h$} \Comment{A chore can be potentially transferred from $b$ to $h$}
            \State $\x_b \gets \x_b \setminus j'$
            \State $\x_h \gets \x_h \cup j'$
        \Else \Comment{No chore can be transferred from $b$ to $h$ or $\ell$}
            \State $\beta \gets \max_{i \in \{\ell, h\}, j \in \x_b} \frac{\alpha_i}{d_i(j)/p_j}$
            \For{$j \in \x_{\ell} \cup \x_h$}
                \State $p_j \gets p_j \cdot \beta$ \Comment{Lower payments until a chore of $b$ is MPB for $\ell$ or $h$}
            \EndFor
        \EndIf
    \Else  \Comment{No chore can be transferred from $b$ or $h$ to $\ell$}
        \State $\beta \gets \max_{j \in \x_b \cup \x_h} \frac{\alpha_{\ell}}{d_{\ell}(j)/p_j}$
            \For{$j \in \x_{\ell}$}
                \State $p_j \gets p_j \cdot \beta$ \Comment{Lower payments until a chore of $b$ or $h$ is MPB for $\ell$}
            \EndFor        
    \EndIf
\EndWhile
\State \Return $\x$
\end{algorithmic}
\end{algorithm}

At a high level, Algorithm~\ref{alg:3agents} maintains a competitive equilibrium where the initially chosen agent $i$ remains the big earner and transfers chores away from $i$. As a result, $i$'s envy towards the other two agents decreases, and eventually, $i$ must cease to be the big earner. We argue that we arrive at an EF1 allocation almost immediately when this happens. While a chore is transferred away from the big earner $b$ (who is $i$ initially and for most of the algorithm) to either the LE $\ell$ or ME $h$, we must carefully perform additional chore transfers between the LE and ME to ensure they remain EF1 w.r.t. each other. To do this, we perform transfers in a very specific manner. 

First, we check if a chore transfer is possible from $b$ to $\ell$ directly (Line 8), and if so, we make the transfer (Lines 9-10). If not, we check if there is a chore $j\in \x_h \cap \mpb_\ell$ that can be potentially transferred from $h$ to $\ell$ (Line 11). If $h$ pEF1-envies $\ell$ w.r.t. a chore $j\in \x_h\cap \mpb_\ell$, then $j$ is transferred from $h$ to $\ell$. If $h$ does not pEF1-envy $\ell$, and if a chore $j'$ can be transferred from $b$ to $h$, then we perform the transfer (Lines 15-17). If there is no such chore $j'$, then the payments of chores owned by both $\ell$ and $h$ are dropped until a chore of $b$ joins one of their MPB sets (Lines 19-21). Finally, if no chore of either $h$ or $b$ can be transferred to, we lower payments of chores of $\ell$ until a chore of $h$ or $b$ joins the MPB set of $\ell$ (Lines 23-25).

Thus, Algorithm~\ref{alg:3agents} makes progress towards obtaining an EF1 allocation by reallocating chores while maintaining the following key properties.
\begin{enumerate}[label=(\roman*)]
    \item Agent $i$ neither gains chores nor experiences payment decreases.
    \item If agent $i$ ceases to be the big earner, an EF1 allocation is found in at most one chore transfer.
    \item The allocation is always fPO. This is done by performing payment decreases as necessary to maintain a competitive equilibrium and ensuring chore transfers do not violate it.
\end{enumerate}
We use Property (i) to bound the number of steps for $i$ to cease being the big earner, after which Property (ii) gives us an EF1 allocation. Property (iii) then guarantees that our EF1 allocation is additionally fPO. We start with two preliminary lemmas.

\begin{lemma}\label{lem:BE_h_LE}
Given an integral competitive equilibrium $(\x,\p)$ with big earner $b$ and least earner $\ell$, for all agents $k \in N$, $\ell$ does not pEF1-envy $k$ and $k$ does not pEF1-envy $b$.
\end{lemma} 
\begin{proof}
Consider an arbitrary agent $k\in N$. We have that $\ell$ does not pEF1-envy $k$ as $\p_{-1}(\x_{\ell}) \leq \p(\x_{\ell}) \leq \p(\x_k)$ by the definition of least earner. Additionally, $k$ does not pEF1-envy $b$ as $\p_{-1}(\x_k) \leq \p_{-1}(\x_b) \leq \p(\x_b)$ by the definition of big earner. 
\end{proof}

\begin{lemma}\label{lem:BEtoLE}
Given an integral competitive equilibrium $(\x, \p)$ with big earner $b$ and least earner $\ell$, if after a single chore transfer $b$ becomes a least earner or $\ell$ becomes a big earner, then the resulting allocation must be pEF1.
\end{lemma}
\begin{proof}
Let $(\x', \p)$ be the CE after the transfer and let $b'$ and $\ell'$ denote the big earner and least earner, respectively, in $(\x',\p)$. 

We first show that if $b = \ell'$ then $(\x', \p)$ is pEF1. Since a chore transfer was performed, $(\x, \p)$ was not pEF1 and thus $\p_{-1}(\x_b) > \p(\x_{\ell})$. Since $b$ becomes a least earner in $(\x', \p)$, it must be that a chore was transferred to $\ell$, as we would otherwise have $\p(\x'_b) \geq \p_{-1}(\x_b) > \p(\x_{\ell}) = \p(\x'_{\ell})$ and $b$ could not be a least earner in $(\x',\p)$. The earning up to one chore of the big earner cannot increase in $(\x',\p)$ as compared to $(\x,\p)$, as $\ell$ is the only agent who gets an extra chore, and we have $\p_{-1}(\x'_{\ell}) \leq \p(\x_{\ell}) < \p_{-1}(\x_b)$. Thus, $\p_{-1}(\x'_{b'} )\leq \p_{-1}(\x_b)$. It follows then that $\p_{-1}(\x'_{b'}) \leq \p_{-1}(\x_b) \leq \p(\x'_b) = \p(\x'_{\ell'})$, since $b=\ell'$. Hence by Lemma~\ref{lem:pEF1impliesEF1} $(\x', \p)$ is pEF1.

We now show that if $\ell = b'$ then $(\x', \p)$ is pEF1. As before, since a chore transfer was performed, $(\x, \p)$ was not pEF1 and so $\p_{-1}(\x_b) > \p(\x_{\ell})$. It must then be that a chore was taken from $b$, as we would otherwise have $\p_{-1}(\x'_b) \geq \p_{-1}(\x_b) > \p(\x_{\ell}) \geq \p_{-1}(\x'_{\ell})$ and $\ell$ could not be a big earner in $(\x'\p)$. The earning of the least earner cannot decrease in $(\x',\p)$ as compared to $(\x,\p)$, as $b$ is the only agent who loses a chore, and we have $\p(\x'_b) \geq \p_{-1}(\x_b) > \p(\x_{\ell})$. Thus, $\p(\x'_{\ell'}) \geq \p(\x_{\ell})$. It follows that $ \p(\x'_{\ell'}) \geq \p(\x_{\ell}) = \p(\x_{b'}) \geq \p_{-1}(\x'_{b'}) $ since $\ell = b'$. Hence by Lemma~\ref{lem:pEF1impliesEF1} $(\x', \p)$ is pEF1.
\end{proof}

We record an important property of the algorithm when the identity of the least earner changes.

\begin{lemma}\label{lem:3agents_LE_change}
During the run of Algorithm~\ref{alg:3agents}, if an agent $\ell$ stops being the least earner, then either the pEF1 condition is satisfied or $\ell$ becomes the middle earner and does not pEF1-envy the new least earner.
\end{lemma}
\begin{proof}
Clearly, it cannot happen that agent $\ell$ stops being the least earner due to a payment decrease. Thus $\ell$ stops being a least earner due to a chore transfer. Let $\x$ be the allocation immediately before the transfer and let $\x'$ be the allocation immediately afterwards, and let $\p$ be the payment vector. Additionally, let $b$ be the big earner and $h$ the middle earner in $(\x, \p)$. By Lemma~\ref{lem:BEtoLE}, if $b$ becomes the least earner or $\ell$ becomes the big earner in $(\x', \p)$, then $(\x', \p)$ is pEF1. Thus, if $(\x',p)$ is not, it must be that in $(\x', \p)$, agent $\ell$ is the new middle earner and $h$ is the new least earner. We now show that the new middle earner $\ell$ does not pEF1-envy the new least earner $h$ by considering the possible combinations of agents involved in the transfer:
\begin{itemize}
\item Suppose a chore was transferred from $b$ to $h$. In this case, $\p(\x'_{\ell}) = \p(\x_{\ell}) < \p_{-1}(\x_b) \leq \p(\x'_b)$, so $\ell$ would remain the least earner in $(\x', \p)$, leading to a contradiction.
\item Suppose a chore was transferred from $b$ to $\ell$. In this case, $\p_{-1}(\x'_{\ell}) \leq \p(\x_{\ell}) \leq \p(\x_h) = \p(\x'_h)$, thus showing that $\ell$ does not pEF1-envy $h$.
\item Suppose a chore was transferred from $h$ to $\ell$. In this case, $\p(\x_h \setminus j) > \p(\x_{\ell})$ for $j\in \x_h$. Then, $\p_{-1}(\x'_{\ell}) \leq \p(\x_{\ell}) < \p(\x_h \setminus j) = \p(\x'_h)$. This again shows that the $\ell$ does not pEF1-envy $h$.
\end{itemize}
In conclusion, either the allocation after the transfer is pEF1, or the new middle earner does not pEF1-envy the new least earner.
\end{proof}

Initially, agent $i$ is allocated all the chores, and is thus the big earner. We show that if agent $i$ stops being the big earner then an EF1+fPO allocation can be found almost immediately.

\begin{lemma}\label{lem:3agentsBE}
If agent $i$ ceases to be the big earner, then Algorithm~\ref{alg:3agents} finds an EF1 allocation in at most one chore transfer.
\end{lemma}
\begin{proof}
We first observe that agent $i$ can only cease to be the big earner if she loses a chore. Indeed, $i$'s chores participate in no payment drops, and chore transfers between middle earner $h$ and least earner $\ell$ only give chores to $\ell$, so only $\ell$ can possibly become the big earner. If this occurs, then by Lemma~\ref{lem:BEtoLE}, the allocation is pEF1 and hence EF1 by Lemma~\ref{lem:pEF1impliesEF1}. Let $\x$ (resp. $\x'$) denote the allocation immediately before (resp. after) $i$ stops being the big earner. By Lemma~\ref{lem:BEtoLE}, if $(\x', \p)$ is not pEF1, it must be that the middle earner $h$ in $(\x, \p)$ becomes the big earner in $(\x', \p)$ while $i$ becomes the middle earner in $(\x, \p)$ and the identity of the least earner $\ell$ remains unchanged. Thus, $i$ first ceases to be the big earner due to the transfer of a chore $j'$ from $i$ to either the middle earner or the least earner agent. We consider these two cases below.

\textbf{Case (1).} Suppose that $i$ first ceases to be the big earner due to the transfer of a chore $j'$ from $i$ to $h$, the middle earner, after which $h$ becomes the new big earner. Then by the design of our algorithm, for an $i$ to $h$ transfer to happen, it must be that there exists $j \in \x_h \cap \s{MPB}_{\ell}$ such that $\p(\x_h \setminus j) \leq \p(\x_{\ell})$. If the new allocation $(\x', \p)$ is not EF1, then Algorithm~\ref{alg:3agents} will immediately transfer $j$ from the new big earner $h$ to least earner $\ell$ resulting in allocation $(\x'', \p)$. We record the relationships between $\x, \x'$ and $\x''$ below:
\begin{equation}\label{eq:allocs}
\x''_i = \x'_i = \x_i \setminus j; \quad \x''_h = \x'_h \setminus j = (\x_h \cup j') \setminus j; \quad \x''_\ell = \x'_\ell \cup j = \x_\ell \cup j.    
\end{equation} 

Due to Lemma~\ref{lem:BEtoLE}, $(\x'', \p)$ is already pEF1 in all cases except the following ones below:
\begin{itemize}
\item $h$ is the BE and $\ell$ is the LE in $(\x'',\p)$. Observe from \eqref{eq:allocs} that $\p_{-1}(\x''_h) \le \p(\x_h\setminus j)$. Since $\p(\x_h\setminus j) \le \p(\x_\ell)$ and $\p(\x_\ell) \le \p(\x''_\ell)$ we obtain  $\p_{-1}(\x''_h) \le \p(\x''_\ell) $, showing that $(\x'',\p)$ is pEF1.
\item $h$ is the BE and $i$ is the LE $(\x'',\p)$. As the case above, we observe that $\p_{-1}(\x''_h) \le \p(\x_\ell) = \p(\x'_\ell)$. Since $i$ is the ME in $(\x', \p)$, we have $\p(\x'_\ell) \le \p(\x'_i)$. Finally since $\p(\x'_i) = \p(\x''_i)$, we obtain $\p_{-1}(\x''_h) \le \p(\x''_i) $, showing that $(\x'',\p)$ is pEF1.
\item $i$ is the BE and $\ell$ is the LE in $(\x'',\p)$. Since $\x''_i = \x'_i$ and $h$ is the BE in $\x'$, we have $\p_{-1}(\x''_i) = \p_{-1}(\x'_i) \le \p_{-1}(\x'_h)$. Since $\x'_h = \x_h \cup j'$, we have $\p_{-1}(\x'_h) \le \p(\x_h)$. Finally since $\p(\x_h\setminus j)\le \p(\x_\ell)$, this implies $\p(\x_h) \le \p(\x_\ell \cup j) = \p(\x''_\ell)$. Putting this all together we obtain $\p_{-1}(\x''_i) \le \p(\x''_\ell)$, implying $(\x'', \p)$ is pEF1.

\end{itemize}



\medskip
\textbf{Case (2).} Suppose that $i$ first ceases to be the big earner due to the transfer of a chore from $i$ to $\ell$. We claim that $\x'$ is then EF1. By Lemma~\ref{lem:BE_h_LE}, no agent pEF1-envies (and thus EF1-envies) the big earner $h$ in $(\x',\p)$. We also see that no agent pEF1-envies $i$ in $(\x',\p)$, since $\p_{-1}(\x'_h) = \p_{-1}(\x_h) \leq \p_{-1}(\x_i) \leq \p(\x'_i)$ by the definition of big earner, and $\p_{-1}(\x'_{\ell}) \leq \p(\x_{\ell}) < \p_{-1}(\x_i) \leq \p(\x'_i)$, since $(\x,\p)$ was not pEF1. It only remains to be shown that $i$ and $h$ do not EF1-envy $\ell$. 

We first show that $h$ does not EF1-envy $\ell$ in $\x'$. Let $t$ be the time-step immediately after $\ell$ most recently became the least earner and let $(\x^t, \p^t)$ denote the allocation and payments at time $t$. At time $t$ it must have been that $i$ was the big earner, $h$ was the middle earner, and $\ell$ was the least earner. By Lemma~\ref{lem:3agents_LE_change} we have that $\p^t_{-1}(\x^t_h) \leq \p^t(\x^t_{\ell})$ and $h$ did not pEF1-envy $\ell$ at time $t$. Note that from time $t$ onwards $\ell$ does not lose any chores, so for $h$ to EF1-envy $\ell$ after time $t$ it must be because $h$ gains some chore $j'$. As seen previously, this only occurs when $h$ already has a chore $j$ which is also MPB for $\ell$, and if $h$ pEF1-envies $\ell$ after receiving $j'$ then Algorithm~\ref{alg:3agents} will immediately transfer $j$ from $h$ to $\ell$ so that $h$ no longer pEF1-envies $\ell$. So, there exists some time $t' \geq t$ such that $h$ does not pEF1-envy $\ell$ and $h$ gains no chores after $t'$. Thus, it cannot be that $h$ EF1-envies $\ell$ in $\x'$.

We now show that $i$ does not EF1-envy $\ell$ in $\x'$. Consider again the time $t'$, at which $h$ does not pEF1-envy $\ell$ and after which does not gain any chores. We therefore have $\p_{-1}(\x'_i) \leq \p_{-1}(\x'_h) \leq \p^{t'}_{-1}(\x^{t'}_h) \leq \p^{t'}(\x^{t'}_{\ell})$. While $i$ has been the big earner, $i$'s chores have never experienced a payment decrease so $i$ has always maintained MPB ratio $\alpha_i = 1$. Then, we have $d_{i_{-1}}(\x'_i) = \p_{-1}(\x'_i) \leq \p^{t'}(\x^{t'}_{\ell}) \leq d_i(\x^{t'}_{\ell})$. Subsequent to $t'$, $\ell$ is the least earner and thus can only have gained chores, so $\x^{t'}_{\ell} \subseteq \x'_{\ell}$ and $d_{i_{-1}}(\x'_i) \leq d_i(\x^{t'}_{\ell}) \leq d_i(\x'_{\ell})$ as desired. Thus, if $i$ first ceases to be the big earner due to the transfer of a chore from $i$ to $\ell$, the resulting allocation is immediately EF1.

In conclusion, if agent $i$ ceases to be the big earner, then Algorithm~\ref{alg:3agents} terminates with an EF1 allocation after at most one subsequent chore transfer.
\end{proof}

We now complete the proof of  Theorem~\ref{thm:3agents}.

\begin{proof}(of Theorem~\ref{thm:3agents})
We show that Algorithm~\ref{alg:3agents} computes an EF1+fPO allocation in strongly polynomial time, which suffices to prove Theorem~\ref{thm:3agents}.

With Lemma~\ref{lem:3agentsBE} in hand, it suffices to bound the number of steps for $i$ can remain the big earner. We have that throughout the algorithm's run, $i$ undergoes no payment drops and may only lose chores. Then, she must cease to be the big earner before losing all her chores, which are at most $m$. 

While $i$ is a big earner, note that there can be at most two payment decreases before a chore transfer must occur: one decrease which makes a chore in $\x_h$ join the MPB set of $\ell$, and one decrease which makes a chore in $\x_i$ join the MPB set for either $h$ or $\ell$. We now bound the number of chore transfers between middle earner $h$, and least earner $\ell$ before a chore must be taken from $i$. Since transfers between $h$ and $\ell$ are always from $h$ to $\ell$, it must be that $h$ becomes the least earner by the time she transfers all of her at most $m$ chores to $\ell$. Suppose $h$ becomes the new least earner after transferring chore $j$ to $\ell$. Let $\x$ be the allocation before transferring $j$ and let $\x'$ be the allocation after transferring $j$. By Lemma~\ref{lem:BEtoLE}, $\ell$ must be the middle earner in $(\x', \p)$ or we already have the pEF1 condition. Since $j$ was transferred from $h$ to $\ell$ we have $\p(\x_h \setminus j) > \p(\x_{\ell})$ and it follows that $\p_{-1}(\x'_{\ell}) \leq \p(\x_{\ell}) < \p(\x_h \setminus j) = \p(\x'_h)$. Notably, it remains that $j$ is MPB for the least earner $h$ in $(\x', \p)$ and the next chore transfer is guaranteed to not be between $h$ and $\ell$. Subsequently, payments of chores in $\x_{\ell} \cup \x_h$ will be decreased until some chore $j'$ in $\x_i$ joins the MPB set for either $\ell$ or $h$, if one is not already. Then $j'$ will be transferred from $i$. Thus, in at most 
$\s{poly}(m)$ chore transfers and payment decreases, a chore is taken from $i$.

Since $i$ can lose only $m$ chores, in $\s{poly}(m)$ steps $i$ ceases to be the big earner. From Lemma~\ref{lem:3agentsBE}, in at most one subsequent step, Algorithm~\ref{alg:3agents} terminates and returns an EF1+fPO allocation.
\end{proof}

Having shown that EF1+fPO allocation is efficiently computable for three agents, a natural follow-up question is to investigate the existence and computation of EFX+fPO allocations for the same class. The following example shows that EFX+fPO allocations need not exist even when there are two agents. 

\begin{theorem}\label{thm:efx-fpo-nonexistence}
There exists a two-type, 2-ary chore allocation instance with two agents that does not admit an EFX+fPO allocation. 
\end{theorem}
\begin{proof}
We consider the following instance with two agents $a$ and $b$ and four chores $j_1,j_2,j_3,j_4$ with disutilities given in the following table:
\begin{center}
\begin{tabular}{|c||c|c|c|c|}
\hline
& $j_1$ & $j_2$ & $j_3$ & $j_4$ \\
\hline
$a$ & 1 & 1 & 3 & 3 \\
\hline
$b$ & 1 & 1 & 4 & 4 \\
\hline
\end{tabular}
\end{center}

Clearly, this instance is a two-type and 2-ary instance with two agents. The only EFX allocation up to swapping of $j_1, j_2$ or $j_3, j_4$ is the allocation $\x$, where $\x_a = \{j_1,j_3\}$ and $\x_b = \{j_2,j_4\}$. We have $d_a(\x_a) = 4$ and $d_b(\x_b) = 5$. Consider the fractional allocation $\y$ where $\y_a = \{j_3, \frac{1}{4}j_4\}$ and $\y_b = \{j_1, j_2, \frac{3}{4}j_4\}$. Then $d_a(\y_a) = 3.75$ but $d_b(\y_b) = 5$, so $\y$ Pareto dominates the EFX allocation $\x$. Thus, this instance admits no EFX+fPO allocation. We note that the same example with $j_1,\dots,j_4$ treated as goods can be used to show the non-existence of EFX+fPO allocations for goods.    
\end{proof}

\section{EF1 + fPO for Two Types of Agents}\label{sec:2types}
In this section, we present Algorithm~\ref{alg:ef1po} which computes an EF1+fPO allocation for two-type instances in strongly polynomial-time. Due to Lemma~\ref{lem:pEF1impliesEF1}, we seek a pEF1 integral competitive equilibrium $(\x,\p)$. Let $N_1$ (resp. $N_2$) be an ordered list of agents with disutility function $d_1$, called Type-1 agents (resp. $d_2$, called Type-2 agents). Algorithm \ref{alg:ef1po} maintains a partition of the chores $M$ into sets $M_1$ and $M_2$, where $M_i$ is allocated to $N_i$, for $i\in\{1,2\}$. Initially, $M_1 = M$ and $M_2 = \emptyset$, and $p_j = d_1(j)$ for each $j\in M$. The chores in $M_1$ are allocated to $N_1$ using the \RR procedure (Algorithm \ref{alg:round-robin}). Given an ordered list of agents $N'$ and chores $M'$, \RR allocates chores as follows. Agents take turns picking chores according to the order specified by $N'$ and each agent picks the least cost chore among the pool of remaining chores during their turn. It is a well-known folklore result that \RR returns an EF1 allocation. 

Initially, agents in $N_1$ have chores while $N_2$ have none, causing agents in $N_1$ to potentially pEF1-envy agents in $N_2$. Thus we must transfer chores from $M_1$ to $M_2$ to reduce this pEF1-envy. When necessary, payments of chores in $M_1$ are raised appropriately before such a transfer to maintain that chores in $M_i$ are MPB for agents in $N_i$, for each $i\in\{1,2\}$\footnote{Note that in this algorithm we perform payment raises instead of payment drops as in Algorithm~\ref{alg:3agents}. This is purely for ease of presentation, as we could equivalently perform payment drops on chores in $M_2$ instead of performing payment raises on chores in $M_1$}. After each transfer, the chores in the (updated) sets $M_i$ are re-allocated to $N_i$ using the \RR procedure, always using the \textit{same ordering} of agents. Since agents in $N_i$ have the same disutility function and chores in $M_i$ are MPB for $N_i$, we have the following feature.

\begin{invariant}\label{inv:ef1po}
Throughout the run of Algorithm~\ref{alg:ef1po}, agents in $N_i$ do not pEF1-envy each other, for each $i\in\{1,2\}$.
\end{invariant}

\begin{algorithm}[t]
\caption{\RR ($\mathsf{RR}$) procedure}\label{alg:round-robin}
\textbf{Input:} Ordered list of agents $N'$, set of chores $M'$\\
\textbf{Output:} An allocation $\x$
\begin{algorithmic}[1]
\State $\x_i \gets \emptyset$ for each $i\in N'$, $i\gets 1$, $P\gets M'$
\While{$P\neq \emptyset$}
\State $j \gets \s{argmin}_{j'\in P} ~d_i(j')$
\State $\x_i \gets \x_i \cup \{j\}$, $P\gets P\setminus \{j\}$, $i\gets i\mod |N'|+1$
\EndWhile
\State \Return $\x$
\end{algorithmic}
\end{algorithm}

By Invariant \ref{inv:ef1po}, no agent can pEF1-envy another agent of the same type. Thus, if we do not have pEF1, it must be that the global BE pEF1-envies the global LE, with the BE and LE being in different groups. Initially, the BE is in $N_1$ and the LE is in $N_2$. Our goal is now to eliminate the pEF1-envy between the BE and LE, and to do this we transfer chores from $M_1$ to $M_2$ with necessary payment raises and \RR re-allocations. We then reconsider the pEF1-envy between the new BE and the new LE. The algorithm terminates when the BE no longer pEF1-envies the LE. We argue that this must happen.

While the BE remains in $N_1$, chores are transferred from $M_1$ to $M_2$. Clearly there can be at most $m$ such transfers, since always $|M_1| \le m$. If in some iteration both the BE and LE belong to the same group $N_i$, then we must be done due to Invariant~\ref{inv:ef1po}. The only remaining case is if the BE becomes an agent in $N_2$ and the LE becomes an agent in $N_1$. We address this case in the following lemma.

\begin{restatable}{lemma}{lemef1podone}\label{lem:ef1po-done}
If the BE is in $N_2$ and the LE is in $N_1$, the allocation must already be pEF1.  
\end{restatable} 
\begin{proof}
We first note that payment-raises do not change the identity of the LE and BE. Therefore, suppose that there is a transfer prior to which the global BE is in $N_1$ and global LE in $N_2$, but after which the global BE is in $N_2$ and global LE in $N_1$. Let $\x$ (resp. $\x'$) denote the allocation immediately before (resp. after) the transfer. For $i\in\{1,2\}$, let $b_i$ and $\ell_i$ denote the BE and LE among agent set $N_i$ before the transfer, and let $b'_i$ and $\ell'_i$ denote the BE and LE among agent sets $N_i$ after the transfer. Let $\p$ be the payments vector accompanying $\x$ and $\x'$. We use the following:

\begin{observation}\label{obs:twotype} \normalfont
In a \RR allocation of a set of chores $M'$ to a list of agents $N' = \{1,\dots, n'\}$ with identical valuations, the big earner (assuming payment vector is proportional to disutility vector) is the agent $i$ who picks the last chore while the least earner is the agent $h = (i\mod n') + 1$ who would pick immediately after $i$.
\end{observation}

\begin{algorithm}[t]
\caption{EF1+fPO for two types of agents}\label{alg:ef1po}
\textbf{Input:} Two-type instance $(N,M,D)$ with $d_i \in \{d_1, d_2\}$\\
\textbf{Output:} An allocation $\x$
\begin{algorithmic}[1]
\State $N_1 \gets \{i \in N \mid d_i = d_1\}$, $N_2 \gets \{i \in N \mid d_i = d_2\}$ 
\State For each $j \in M$, set $p_j \gets d_1(j)$
\State $M_1 \gets M, M_2 \gets \varnothing$
\State $\x \gets \mathsf{RR}(N_1, M_1) \cup \mathsf{RR}(N_2, M_2)$ \Comment{$\mathsf{RR}$ = \RR}
\While{$(\x,\p)$ is not pEF1} 
	\State $\mpb_2 \gets \{j \in M : j \text{ is MPB for agents in } N_2\}$
	\If{$\exists j \in M_1 \cap \mpb_2$}
		\State $M_1 \gets M_1 \setminus \{j\}$, $M_2 \gets M_2 \cup \{j\}$
		\State $\x \gets \mathsf{RR}(N_1, M_1) \cup \mathsf{RR}(N_2, M_2)$ 
	\Else
	\State Raise payments of $M_1$ until $|M_1 \cap \mpb_2| > 0$
	\EndIf
\EndWhile
\State \Return $\x$
\end{algorithmic}
\end{algorithm}

This is because agents pick chores according to increasing disutility since they all have the same cost function. Thus in $\x$, the BE $b_i$ picked the last chore when $M_i$ was \RR-allocated to $N_i$, for $i\in\{1,2\}$. We now examine how the identity and earning of the BE and LE of each agent set $N_i$ change after a chore $j$ is transferred from $M_1$ to $M_2$. 

When $M_1\setminus\{j\}$ is \RR-allocated to $N_1$, the agent who picks immediately before $b_1$ now picks last, and is the new BE. Thus Obs.~\ref{obs:twotype} implies that new LE in $N_1$ is in fact $b_1$, i.e., $\ell_1' = b_1$. Additionally, $b_1$'s new total earning is at least as much as her previous earning without her highest paying chore, since in each round up to the last round she must now pick a weakly higher disutility (and thus higher paying) chore than before, but she no longer picks a chore in the last round. Thus:
\begin{equation}\label{eq:earning1}
\p(\x'_{\ell'_1}) = \p(\x'_{b_1}) \ge \p_{-1}(\x_{b_1}). 
\end{equation}

Conversely, when $M_2\cup \{j\}$ is \RR-allocated to $N_2$, the previous LE $\ell_2$ now picks the last chore. Thus by Obs.\ref{obs:twotype}, $\ell_2$ is the new BE, i.e., $b_2' = \ell_2$. In each round up to the last $\ell_2$ now picks a weakly lower disutility (and thus lower paying) chore than before, but $\ell_2$ now receives a new worst chore in the last round. Thus:
\begin{equation}\label{eq:earning2}
\p_{-1}(\x'_{b'_2}) = \p_{-1}(\x'_{\ell_2})  \leq \p(\x_{\ell_2}).    
\end{equation}

Let us now examine the pEF1-envy before and after the chore transfer. Prior to the transfer, i.e., in $(\x,\p)$, the global BE and LE are $b_1$ and $\ell_2$ respectively. Since $(\x,\p)$ is not pEF1, $\p_{-1}(\x_{b_1}) > \p(\x_{\ell_2})$. Using \eqref{eq:earning1} and \eqref{eq:earning2}, we get:
\[ \p(\x'_{\ell'_1}) \ge \p_{-1}(\x_{b_1}) >  \p(\x_{\ell_2}) \ge  \p_{-1}(\x'_{b'_2}),\]
implying that after the transfer, i.e., in $(\x',\p)$, the global BE $b'_2$ does not pEF1-envy the global LE $\ell'_1$. Thus $(\x',\p)$ must be pEF1 by Lemma~\ref{lem:BEenvyLE}.
\end{proof}

In conclusion, chores are transferred unilaterally from $M_1$ to $M_2$ with necessary payment-raises and \RR re-allocations among agents of the same type until the allocation is pEF1+fPO. Clearly, Algorithm~\ref{alg:ef1po} runs in $\poly{n,m}$ time. We conclude:

\begin{theorem}\label{thm:2types}
Given a two-type chore allocation instance $(N,M,D)$, an EF1 + fPO allocation exists and can be computed in strongly polynomial-time.
\end{theorem}

In contrast, Theorem~\ref{thm:efx-fpo-nonexistence} shows that EFX+fPO allocations need not exist for two-type instances. We conclude this section by noting that the same techniques can be applied for goods.
\begin{remark}\label{rem:2types}
An EF1+fPO allocation is strongly polynomial-time computable for a two-type goods allocation instance.
\end{remark}

\section{EFX + fPO for Three Bivalued Agents}\label{sec:bivalued3agents}

Before stating our third result, we recall from Theorem~\ref{thm:efx-fpo-nonexistence} that an EFX+fPO allocation is not guaranteed to exist, even for 2-ary instances. We therefore study the computation of EFX+fPO allocations for bivalued instances. Our third result is:

\begin{restatable}{theorem}{mainclaimone}\label{thm:bivalued3agents}
Given a bivalued chore allocation instance $(N,M,D)$ with three agents, an EFX+fPO allocation exists and can be computed in strongly polynomial-time.
\end{restatable}

We prove Theorem~\ref{thm:bivalued3agents} by showing that Algorithm~\ref{alg:efxpo} computes an EFX+fPO allocation in polynomial-time. We remark that the algorithm of \cite{zhou22efx} for three bivalued agents returns an allocation that is EFX, but not necessarily PO (see Ex.~\ref{ex:efx-alg-not-po}). We present a complete description and analysis of Algorithm~\ref{alg:efxpo} in Appendix~\ref{app:3bivaluedagents} and examples illustrating its execution in Appendix \ref{app:examples}.

\subsection{Additional Notation}\label{sec:addNot}

We first note that given a bivalued instance, we may assume w.l.o.g. that all $d_i(j) \in \{1, k\}$ by re-scaling the valuations. We may also assume that for every agent $i$ there exists some chore $j$ such that $d_i(j) = 1$. Otherwise, if $d_i(j) = k$ for all $j \in M$, re-scaling agent $i$'s valuations gives $d_i(j) = 1$ for all $j \in M$. We then partition the set of chores $M$ into the set of \textit{L-chores} $L$ and the set of \textit{K-chores} $K$:
\begin{itemize}
	\item $L = \{j \in M \mid \exists i \in N \text{ s.t. } d_i(j) = 1\}$, and
	\item $K = \{j \in M \mid \forall i \in N, d_i(j) = k\}$.
\end{itemize}
Further, for each agent $i$ we partition her bundle $\x_i$ into her set of \textit{1-chores} $\x_{i_1}$ and \textit{k-chores} $\x_{i_k}$:
\begin{itemize}
	\item $\x_{i_1} = \{j \in \x_i \mid d_i(j) = 1\}$, and
	\item $\x_{i_k} = \{j \in \x_i \mid d_i(j) = k\}$.
\end{itemize}

We highlight the difference between a $k$-chore and a $K$-chore. A $k$-chore need only give disutility $k$ to the agent it is allocated to, while a $K$-chore gives disutility $k$ to all agents. A $K$-chore will be a $k$-chore for any agent it is allocated to, but an agent's $k$-chore need not be a $K$-chore.

\subsection{Overview of Algorithm~\ref{alg:BalancedChores}}\label{sec:overviewsimple}

\begin{restatable}{algorithm}{balanced}
\caption{Computing a balanced EF1+fPO allocation}\label{alg:BalancedChores}
\textbf{Input:} Chores instance $(N,M,D)$ with $d_{i}(j)\in \{1,k\}$\\
\textbf{Output:} Allocation $\x$, payments $\p$, agent groups $\{N_r\}_{r\in [R]}$
\begin{algorithmic}[1]
\State $(\x,\p), \{N_r\}_{r \in [R]} \gets \s{MakeInitGroups}(N, L, D)$
\For{$j \in K$}
\State $\ell \gets \s{argmin}_{i \in N} |\x_i|$, ties broken by lowest group, then fewest $K$-chores
\State $\x_\ell \gets \x_\ell \cup \{j\}, \text{ } p_j = k$
\EndFor
\State $a \gets \s{argmax}_{i \in N} |\x_i|$, ties broken by highest group
\State $\ell \gets \s{argmin}_{i \in N} |\x_i|$, ties broken by lowest group, then fewest $k$-chores
\While{$|\x_\ell| < |\x_a| - 1$}
\If{$\exists j \in \x_a \mid j \in \mpb_\ell$}
\State $\x_a \gets \x_a \setminus \{j\}$, \text{ } $\x_\ell \gets \x_\ell \cup \{j\}$
\State $a \gets \s{argmax}_{i \in N} |\x_i|$, tiebreaks as in Line 5
\State $\ell \gets \s{argmin}_{i \in N} |\x_i|$, tiebreaks as in Line 6
\Else
\State Raise payments in $a$'s group by a factor of $k$
\EndIf
\EndWhile
\State \Return $(\x,\p, \{N_r\}_{r\in [R]})$
\end{algorithmic}
\end{restatable}

Algorithm~\ref{alg:efxpo} first uses Algorithm~\ref{alg:BalancedChores} to obtain an initial competitive equilibrium which is EF1+fPO. Unlike previous algorithms for bivalued chores, Algorithm~\ref{alg:BalancedChores} establishes the EF1 condition not by payment envy-freeness, but by \textit{balancing} the number of chores of the agents. We say that an allocation $\x$ has \textit{balanced total chores}, or simply \textit{balanced chores}, if $\max_{i \in N} |\x_i| - \min_{i \in N} |\x_i| \leq 1$.

Algorithm~\ref{alg:BalancedChores} begins by calling subroutine $\s{MakeInitGroups}$ which appears in \cite{Garg_Murhekar_Qin_2022} (pseudocode is included in the Appendix~\ref{app:3bivaluedagents} for completeness). It returns a competitive equilibrium $(\x, \p)$ of the set of $L$-chores $L$ which is cost-minimizing, so that $\sum_{i \in N} d_i(\x_i)$ is minimized over all allocations, as well as a partition of the agents into $R$ ordered groups $\{N_r\}_{r \in [R]}$. We say that agent group $N_r$ is higher (resp. lower) than group $N_{r'}$ if $r < r'$ (resp. $r > r'$). We now record several properties of the groups:
\begin{enumerate}[label=(\roman*)]
	\item Let $f(r) = \max_{i \in N_r} |\x_i|$. $f(r)$ weakly decreases with $r$. 
	\item Agents within a group have balanced total chores.
	\item For $h \in N_r$, $i \in N_{r'}$ with $r < r'$, $\forall j \in \x_h$, $d_i(j) = k$.
\end{enumerate}

After allocating $L$ and obtaining the agent groups, we allocate the set of $K$-chores $K$ by giving chores to the agent $\ell$ with the fewest number of total chores (with ties broken by lowest group, then fewest $K$-chores, then arbitrarily) in the hopes of obtaining an allocation with balanced total chores. If after allocating $K$ in this fashion $\x$ is still not balanced, we balance the chores by transferring chores from the agent $a$ with the most total chores (ties broken by highest group, then arbitrarily) to the agent $\ell$ with the fewest total chores (ties broken by lowest group, then fewest $k$-chores, then arbitrarily), with payments of chores in $a$'s group raised as necessary to maintain that all agents have only MPB chores. 

\begin{restatable}{lemma}{BCTerm}\label{lem:BC_Term} Algorithm~\ref{alg:BalancedChores} terminates in polynomial-time with a balanced allocation which is EF1+fPO.
\end{restatable}

\subsection{Overview of Algorithm~\ref{alg:efxpo}}\label{sec:overview3agents}

\begin{restatable}{algorithm}{algefxfpo}
\caption{Computing an EFX+fPO allocation}\label{alg:efxpo}
\textbf{Input:} Bivalued Chores instance $(N,M,D)$ with $|N|=3$\\
\textbf{Output:} An allocation $\x$
\begin{algorithmic}[1]
\State $(\x, \p, \{N_r\}_{r \in [R]}) \gets \s{Algorithm~\ref{alg:BalancedChores}(N,M,D)}$ 
\If{$R = 1$}
	\State $(\x, \p) \gets \s{Algorithm~\ref{alg:R1}}(\x, \p)$ 
\ElsIf{$R = 2$ and $|N_1| = 1$}
	\State $(\x, \p) \gets \s{Algorithm~\ref{alg:R2-1}}(\x, \p)$
\ElsIf{$R = 2$ and $|N_1| = 2$}
	\State $(\x, \p) \gets \s{Algorithm~\ref{alg:R2-2}}(\x, \p)$
\EndIf
\State \Return $\x$
\end{algorithmic}    
\end{restatable}

Algorithm~\ref{alg:efxpo} begins with the balanced EF1+fPO allocation returned by Algorithm~\ref{alg:BalancedChores}. Then, based on the structure of the agent groups according to Algorithm~\ref{alg:BalancedChores}, we perform chore transfers to improve EF1 to EFX while maintaining the fPO condition. Furthermore, we will assume that the three agents do not have identical disutilities\footnote{\label{idValBivalued} 
An EFX+fPO allocation is obtainable for agents with identical disutilities via the well-known envy-cycle procedure of \cite{lipton} by allocating chores in order of decreasing disutility.}. 
Recall that $R$ denotes the number of agent groups returned by Algorithm~\ref{alg:BalancedChores}. 

We first consider the case when $R=1$ and Algorithm~\ref{alg:BalancedChores} returns one agent group. We show that Algorithm~\ref{alg:BalancedChores} returns a cost-minimizing allocation $\x$ so that $\forall i \in N, \x_{i_k} \subseteq K$. We then show that $\x$ has \textit{balanced k-chores}, i.e, $\max_{i \in N} |\x_{i_k}| - \min_{i \in N} |\x_{i_k}| \leq 1$, and EFX-envy can only exist from an agent with more $K$-chores towards an agent with fewer $K$-chores. For three agents, this implies that at most two EFX-envy relationships may exist. Algorithm~\ref{alg:R1} then finds an allocation which maintains cost-minimization and balanced $k$-chores while guaranteeing that at most one EFX-envy relationship exists. If it does, we show that there exists an agent $i$ with more $K$-chores than another agent $h$ and a chore $j \in \x_h$ such that $j \notin \mpb_i$. Algorithms \ref{alg:R1-2extra} and \ref{alg:R1-1extra} use these properties to find an EFX 
allocation in at most $m$ transfers.

We next consider when $R=2$ and Algorithm~\ref{alg:BalancedChores} returns two groups. Then, Algorithm~\ref{alg:R2-1} and Algorithm~\ref{alg:R2-2} leverage the relationship between groups outlined in Property (iii) of Section~\ref{sec:overviewsimple}. We use this in conjunction with the fact that Algorithm~\ref{alg:BalancedChores} favors giving $k$-chores to agents in lower groups to show that EFX-envy can only exist between agents in the same group. Then, Property (iii) allows  us to find an EFX allocation using only constantly many chore transfers.

Finally, when $R=3$, each agent belongs to her own group. In this case, we show that the allocation returned by Algorithm~\ref{alg:BalancedChores} is already EFX.

As these cases are exhaustive, we conclude that Algorithm~\ref{alg:efxpo} returns an EFX+fPO allocation in polynomial time, 
thus proving Theorem~\ref{thm:bivalued3agents}.

\section{Conclusion} 
In this work, we provided improved guarantees for fair and efficient allocation of chores, under the fairness notions of EF1/EFX, and the efficiency notion of fPO. Our algorithms for the three agents and two-types setting are among the few positive non-trivial results known for the EF1+fPO problem, in addition to that of bivalued chores. Combining CE-based frameworks with envy-resolving algorithms like \RR may be an important tool in settling the problem in its full generality. We also made progress on computing an EFX+fPO allocation for three bivalued agents. Extending and generalizing our approach to higher numbers of agents is a natural next step.

\bibliography{references}

\newpage

\appendix

\section{EFX+fPO for Three Bivalued Agents}\label{app:3bivaluedagents}

In this section, we present a self-contained exposition of Section~\ref{sec:bivalued3agents}. We prove the following result:
\mainclaimone*

We prove Theorem~\ref{thm:bivalued3agents} by showing that Algorithm~\ref{alg:efxpo} computes an EFX+fPO allocation in polynomial-time. Note that Algorithm~\ref{alg:efxpo} begins with the balanced EF1+fPO allocation returned by Algorithm~\ref{alg:BalancedChores}.
\setcounter{algorithm}{3}
\balanced

Based on the structure of the agent groups according to Algorithm~\ref{alg:BalancedChores}, Algorithm~\ref{alg:efxpo} performs chore transfers to improve EF1 to EFX while maintaining the fPO condition. Let $R$ denote the number of agent groups returned by Algorithm~\ref{alg:BalancedChores}.

\algefxfpo

We first consider the case when $R=1$ and Algorithm~\ref{alg:BalancedChores} returns one agent group. We show that Algorithm~\ref{alg:BalancedChores} returns a cost-minimizing allocation $\x$ so that $\forall i \in N, \x_{i_k} \subseteq K$. We then show that $\x$ has \textit{balanced k-chores}, i.e, $\max_{i \in N} |\x_{i_k}| - \min_{i \in N} |\x_{i_k}| \leq 1$, and EFX-envy can only exist from an agent with more $K$-chores towards an agent with fewer $K$-chores. For three agents, this implies that at most two EFX-envy relationships may exist. Algorithm~\ref{alg:R1} then finds an allocation which maintains cost-minimization and balanced $k$-chores while guaranteeing that at most one EFX-envy relationship exists. If it does, we show that there exists an agent $i$ with more $K$-chores than another agent $h$ and a chore $j \in \x_h$ such that $j \notin \mpb_i$. Algorithms \ref{alg:R1-2extra} and \ref{alg:R1-1extra} use these properties to find an EFX allocation in at most $m$ transfers.

We next consider when $R=2$ and Algorithm~\ref{alg:BalancedChores} returns two groups. Then, Algorithm~\ref{alg:R2-1} and Algorithm~\ref{alg:R2-2} leverage the relationship between groups outlined in Property (iii) of Section~\ref{sec:overviewsimple}. We use this in conjunction with the fact that Algorithm~\ref{alg:BalancedChores} favors giving $k$-chores to agents in lower groups to show that EFX-envy can only exist between agents in the same group. Then, Property (iii) allows  us to find an EFX allocation using only constantly many chore transfers.

Finally, when $R=3$, each agent belongs to her own group. In this case, we show that the allocation returned by Algorithm~\ref{alg:BalancedChores} is already EFX.

As these cases are exhaustive, we conclude that Algorithm~\ref{alg:efxpo} returns an EFX+fPO allocation in polynomial time, 
thus proving Theorem~\ref{thm:bivalued3agents}. We now discuss the subroutines of Algorithm~\ref{alg:efxpo} and their properties in detail.

\subsection{Additional Preliminaries}\label{sec:MakeInitGroups}
We first restate some preliminaries from Garg, Murhekar, and Qin~\cite{Garg_Murhekar_Qin_2022}.

\begin{algorithm}[!b]
\caption{$\s{MakeInitGroups}$}\label{alg:MakeInitGroups}
\textbf{Input:} Chore allocation instance $(N,M,D)$ with $d_i(j) \in \{1,k\}$\\
\textbf{Output:} Competitive equilibrium $(\x, \p)$, agent groups $\{N_r\}_{r\in [R]}$
\begin{algorithmic}[1]
\State $(\x,\p) \gets$ initial cost minimizing allocation, where $p_j = d_i(j)$ for $j \in \x_i$.
\State $R\gets 1$, $N' \gets N$
\While{$N' \neq \emptyset$}
\State $b \gets \s{argmax}_{i \in N'} \p_{-1}(\x_i)$ \Comment{Big Earner}
\State $C_b \gets$ Component of $b$ 
\NoDo
\While{$\exists$ agent $i \in C_b$ s.t. $\p_{-1}(\x_b) > \p(\x_i)$}
\ReDo
\State Let $(b, j_1, h_1, j_2, \dots, h_{\ell - 1}, j_\ell, i)$ be the shortest \text{ \qquad\,\, } alternating path from $b$ to $i$
\State $\x_{h_{\ell - 1}} \gets \x_{h_{\ell - 1}} \setminus \{j_\ell\}$ \Comment{Chore transfer}
\State $\x_i \gets \x_i \cup \{j_\ell\}$
\State $b \gets \s{argmax}_{i\in N'} \p_{-1}(\x_i)$
\EndWhile  
\State $H_R \gets C_b \cap (N' \cup \x_{N'})$ 
\State $N_R \gets H_R \cap N$ \Comment{Agent group}
\State $N' \gets N' \setminus N_R$, $R\gets R+1$
\EndWhile
\State \Return $(\x,\p, \{N_r\}_{r\in [R]})$
\end{algorithmic}
\end{algorithm}

Given a chore allocation instance $(N, M, D)$ and competitive equilibrium $(\x, \p)$, we define the MPB graph to be a bipartite graph $G = (N, M, E)$ where for an agent $i$ and chore $j$, $(i, j) \in E$ if and only if $j \in \s{MPB}_i$. Further, an edge $(i, j)$ is called an \textit{allocation edge} if $j \in \x_i$, otherwise it is called an \textit{MPB} edge. 

For agents $i_0,\dots,i_{\ell}$ and chores $j_1,\dots,j_{\ell}$, a path $P = (i_0, j_1, i_1, j_2, \dots, j_{\ell}, i_{\ell})$ in the MPB graph, where for all $1\le \ell' \le \ell$, $j_{\ell'} \in \x_{i_{\ell'-1}}\cap \s{MPB}_{i_{\ell'}}$, is called a \textit{special path}. We define the \textit{level} $\lambda(h;i_0)$ of an agent $h$ w.r.t. $i_0$ to be half the length of the shortest special path from $i_0$ to $h$, and to be $n$ if no such path exists. A path $P = (i_0, j_1, i_1, j_2, \dots, j_{\ell}, i_{\ell})$ is an \textit{alternating path} if it is special, and if $\lambda(i_0;i_0) < \lambda(i_1;i_0) \dots < \lambda(i_{\ell};i_0)$, i.e., the path visits agents in increasing order of their level w.r.t. $i_0$. Further, the edges in an alternating path alternate between allocation edges and MPB edges. Typically, we consider alternating paths starting from a big earner agent. For a big earner $i$, define $C_i^{\ell}$ to be the set of all chores and agents which lie on alternating paths of length $\ell$. Call $C_i = \bigcup_{\ell} C_i^{\ell}$ the \textit{component} of $i$, the set of all chores and agents reachable from the big earner $i$ through alternating paths.

We now expand our definitions of chore balance. Given a set of agents $A$ and an allocation $\x$, we say that agents in $A$ have
\begin{enumerate}[label=(\arabic*)]
    \item \textit{balanced 1-chores} if $\max_{i \in A} |\x_{i_1}| - \min_{i \in A} |\x_{i_1}| \leq 1$,
    \item \textit{balanced k-chores} if $\max_{i \in A} |\x_{i_k}| - \min_{i \in A} |\x_{i_k}| \leq 1$,
    \item \textit{balanced total chores}, or just simply \textit{balanced chores}, if $\max_{i \in A} |\x_i| - \min_{i \in A} |\x_i| \leq 1$, and
    \item \textit{fully balanced chores} if (1), (2), and (3) hold.
\end{enumerate}
Specifically, when $A = N$ we refer to the allocation $\x$ rather than $N$, e.g., ``allocation $\x$ has balanced chores."

\subsection{Algorithm~\ref{alg:BalancedChores}: Balanced EF1+fPO}\label{sec:proofsBalancedChores}

Recall that Algorithm~\ref{alg:BalancedChores} obtains a partition of the agents into groups $\{N_r\}_{r \in [R]}$. These groups maintain a number of features throughout the run of Algorithm~\ref{alg:BalancedChores}.
\begin{invariant}\label{inv:Almost_Fewer_Chores} \normalfont
For agents $i \in N_r$, $h \in N_{r'}$ with $r < r'$, $|\x_h| \leq |\x_i| + 1$.
\end{invariant}
The allocation returned by $\s{MakeInitGroups}$ in Line 1 is guaranteed to have this feature due to Properties (i) and (ii) from Section~\ref{sec:overviewsimple}. Then, in the allocation of chore set $K$ in Lines 2-4, $h$ may only gain a chore if $|\x_h| \leq |\x_i|$, so Invariant~\ref{inv:Almost_Fewer_Chores} holds after $h$ receives a $K$-chore. We show that Invariant~\ref{inv:Almost_Fewer_Chores} holds after any chore transfer by induction. Suppose that Invariant~\ref{inv:Almost_Fewer_Chores} holds before a chore transfer. If an agent $i$ does not lose a chore in the next transfer, the previous argument suffices. If $i$ does lose a chore, then $|\x_i| \geq |\x_h|$ before the transfer and Invariant~\ref{inv:Almost_Fewer_Chores} can only be violated if $h$ is the receiver of $i$'s chore. In this case, the loop condition in Line 7 gives $|\x_h| < |\x_{i^*}| - 1$ before the transfer so Invariant~\ref{inv:Almost_Fewer_Chores} holds for $i^*$ and $h$ after the transfer and we are done.

\begin{invariant}\label{inv:Group_Balance} \normalfont
All groups maintain balanced chores.
\end{invariant}
By Property (ii) of Section~\ref{sec:overviewsimple}, the groups returned by $\s{MakeInitGroups}$ have balanced chores. Property (ii) is also clearly maintained through the allocation of set $K$ as chores are only given to agents who have fewest chores. We now show that all groups have balanced chores after a chore transfer by induction. Suppose that all groups have balanced chores before a chore transfer. Then a chore cannot be transferred between agents in the same group. By Invariant~\ref{inv:Almost_Fewer_Chores}, a chore cannot be transferred from a lower group to a higher group. It must be then that a chore is transferred from a higher group to a lower group. Since the giver of the chore $a$ is an agent with maximum chores, $a$'s group must remain balanced after her loss. Similarly, since the receiver of the chore $\ell$ is an agent with minimum chores, $\ell$'s group must also remain balanced. Thus, all groups maintain balanced chores after the transfer.

\begin{observation}\label{obs:High_To_Low} \normalfont
If a chore is transferred from agent $a$ to agent $\ell$, then $a \in N_r$ and $\ell \in N_{r'}$ such that $r < r'$. 
\end{observation}
By Invariant~\ref{inv:Almost_Fewer_Chores} it cannot be that $a$ is in a lower group than $\ell$. By Invariant~\ref{inv:Group_Balance}, it cannot be that $a$ is  in the same group as $\ell$. Thus, it must be that $a$ is in a higher group than $\ell$. 

\begin{observation}\label{obs:NondecreasingMinChores} \normalfont
Let $\ell$ be an agent with the fewest chores. Then $|\x_\ell|$ is non-decreasing throughout the run of Algorithm~\ref{alg:BalancedChores}.
\end{observation}
This clearly holds during the allocation of $K$ in Lines 2-4. Suppose $\ell$ has $q$ chores at time $t$ before a chore transfer. An agent $a$ can only lose a chore if she has at least $q + 2$ chores, so after the transfer every agent continues to have at least $q$ chores.

\begin{observation}\label{obs:NonincreasingMaxChores} \normalfont
Let $a$ be an agent with the most chores. Then $|\x_a|$ is non-increasing after Line 7.
\end{observation}
Suppose $a$ has $q$ chores at time $t$ before a chore transfer. An agent $\ell$ can only gain a chore if she has at at most $q - 2$ chores, so after the transfer every agent continues to have at most $q$ chores.

\begin{restatable}{lemma}{Gain_Or_lose}\label{lem:Gain_Or_Lose}
After calling $\s{MakeInitGroups}$, if an agent in $N_r$ is chosen to gain (resp. lose) a chore, then no agent in $N_r$ can ever be chosen to lose (resp. gain) a chore. 
\end{restatable}

\begin{proof}
We first note that by Invariant~\ref{inv:Group_Balance} it cannot be that agents in $N_r$ are chosen to gain and lose a chore in the same loop iteration.

Suppose an agent in $N_r$ is first chosen to gain a chore and an agent $N_r$ is later chosen to lose a chore. Let $i$ be the first agent in $N_r$ chosen to lose a chore at time $T$, and let $h$ be the last agent in $N_r$ chosen to gain a chore before $T$ at time $T'$. For $i$ to be chosen to lose a chore, it must be that $h$ successfully gained a chore, i.e., Lines 7-13 did not loop infinitely with payment raises. Suppose $h$ had $q$ chores at time $T'$ before gaining a chore. Then by Observation~\ref{obs:NondecreasingMinChores}, after $T'$ every agent has at least $q$ chores. By Invariant~\ref{inv:Group_Balance}, every agent in $N_r$ has either $q$ or $q+1$ chores after $T'$. Then, by selection of $T'$ and $T'$, no agent in $N_r$ can gain or lose chores between $T'$ and $T$. Thus, at time $T$ when $i$ is chosen to lose a chore, $i$ has at most $q+1$ chores. Yet, since every agent has at least $q$ chores at time $T$, $i$ cannot be chosen to lose a chore as this implies the allocation is already balanced. 

An equivalent proof shows the reverse, that an agent in $N_r$ cannot be first chosen to lose a chore with another agent in $N_r$ subsequently chosen to gain a chore.
\end{proof}

Notably Lemma~\ref{lem:Gain_Or_Lose} gives the following weaker condition.

\begin{corollary}\label{cor:Gain_Or_Lose}
After calling $\s{MakeInitGroups}$, if an agent in $N_r$ has gained (resp. lost) a chore, then no agent in $N_r$ can ever have lost (resp. gained) a chore.
\end{corollary}

We now show that the following conditions hold in Algorithm~\ref{alg:BalancedChores}.

\begin{restatable}{lemma}{Payment_Raises}\label{lem:Payment_Raises} In Algorithm~\ref{alg:BalancedChores}:
\begin{enumerate}
    \item Each group has payments raised at most once.
    \item Groups have payments raised in increasing order of $r$.
    \item If a group has lost a chore, it is a raised group.
    \item Agents in raised groups have MPB ratio $1/k$ while agents in unraised groups have MPB ratio $1$.
    \item Chores belonging to agents in raised groups are MPB for agents in unraised groups.
\end{enumerate}
\end{restatable}
\begin{proof}
We show that Conditions (1), (2), (3), (4), and (5) hold by induction. Before any chore transfer or payment raise, Conditions (1), (2), (3), and (5) are vacuously true. Condition (4) is trivially seen. Suppose that all conditions hold up to time $T$ when Algorithm~\ref{alg:BalancedChores} has chosen an agent $a$ to lose a chore and an agent $\ell$ to gain a chore. We show that they all hold after a chore transfer or payment raise. 

Suppose $\exists j \in \x_a \cap \mpb_\ell$ so a chore transfer occurs at $T$. Clearly Conditions (1), (2), and (3) will hold after the transfer. It suffices to show that $a$ belongs to a raised group to show Conditions (4) and (5). Suppose that $a$ does not belong to a raised group. By Invariant~\ref{inv:Almost_Fewer_Chores}, it cannot be that $\ell$ belongs to a group above $a$'s group. By Invariant~\ref{inv:Group_Balance}, it cannot be that $\ell$ belongs to $a$'s group. It must be then that $\ell$ belongs to an unraised group below $a$'s group. By Lemma~\ref{lem:Gain_Or_Lose}, it cannot be that $a$'s group has gained a chore and it cannot be that $\ell$'s group has lost a chore. Then $a$'s group has only chores allocated to them by $\s{MakeInitGroups}$. Since $\ell$ belongs to a group below $a$'s, by Property (iii) of Section~\ref{sec:overviewsimple} all of $a$'s chores are disutility $k$ for $\ell$. Further, as $a$'s group is unraised, it remains that these chores have payment 1. Then, $\forall j \in \x_a$, $\alpha_{\ell j} = d_\ell(j) / p_j = k$. Since $\ell$ is in an unraised group, by Condition (4) $\alpha_\ell = 1$. But then $\nexists j \in \x_a \cap \mpb_\ell$, a contradiction. So it must be that $a$ in fact belongs to a raised group and all conditions hold. 

Now suppose then that $\nexists j \in \x_a \cap \mpb_\ell$ so $a$'s group is to go undergo a payment raise at time $T$. Let $N_1, \ldots N_r$ be the raised groups up to time $T$. We show that $a \in N_{r+1}$. Suppose rather that $a$ belongs to a raised group. Note that any raised group must contain an agent who was previously selected to lose a chore. Then since $\ell$ is chosen to gain a chore, by Lemma~\ref{lem:Gain_Or_Lose} $\ell$ cannot belong to a raised group. Yet, if $\ell$ belongs to an unraised group then by Condition (4) $\exists j \in \x_a \cap \mpb_\ell$, a contradiction. So it must be that $a$ belongs to an unraised group. Then by Condition (3) $a$'s group has not previously lost a chore. Since $a$ has been chosen to lose a chore, by Lemma~\ref{lem:Gain_Or_Lose} $a$'s group also cannot have gained a chore. Thus, $a$ must belong to a group which has neither gained nor lost a chore since $\s{MakeInitGroups}$ was called. Then among such groups, Property (i) of Section~\ref{sec:overviewsimple} and the tiebreaking rules guarantee that $a \in N_{r+1}$. It is clear that Conditions (1), (2), and (3) hold after payments of $N_{r+1}$ are raised. By Condition (4) agents in $N_{r+1}$ had MPB ratio 1 before the payment raise. Then clearly they have MPB ratio $1/k$ afterwards. Since all chores in $N_{r+1}$ have payment $k$ and give at least disutility 1 to all agents, it must be that agents in other raised groups also continue to have MPB ratio $1/k$. Finally, all unraised groups are below $N_{r+1}$ ($a$'s group) and we have seen that all chores belonging to agents in $N_{r+1}$ are disutility $k$ for agents in lower groups. Then for all $j$ belonging to an agent in $N_{r+1}$ and any unraised agent $i$, $\alpha_{ij} = k/k = 1$. This shows Conditions (4) and (5).

Thus, all conditions hold after either a chore transfer or payment raise.
\end{proof}

We note the following about chore gains.
\begin{restatable}{lemma}{Gain_k}\label{lem:Gain_k}
If an agent in group $N_r$ gains a chore $j$ after $\s{MakeInitGroups}$, then $\forall h \in \bigcup^R_{i=r} N_r, d_h(j) = k$.
\end{restatable}
\begin{proof}
It is either the case that $j \in K$ or $j$ is transferred to an agent in $N_r$ from some higher group $N_{r^*}$ (Observation~\ref{obs:High_To_Low}). Lemma~\ref{lem:Gain_k} is clear in the former case. In the latter case, from Corollary~\ref{cor:Gain_Or_Lose} it must be that agents in $N_{r^*}$ only lose chores. Then the only chores of agents in $N_{r^*}$ are those allocated to them in $\s{MakeInitGroups}$. Then the result follows by Property (iii) of Section~\ref{sec:overviewsimple}. 
\end{proof}

We then also have the following.
\begin{restatable}{lemma}{Higher_k}\label{lem:Higher_k}
Given $i \in N_r$ and $h \in N_{r'}$ with $r < r'$, for all $j \in \x_i$, $d_h(j) = k$.
\end{restatable}
\begin{proof}
For all $j \in \x_i$, $j$ must be:
\begin{itemize}
    \item a chore allocated to $i$ in $\s{MakeInitGroups}$. Then Property (iii) of Section~\ref{sec:overviewsimple} gives $d_h(j) = k$.
    \item a $K$-chore. Then $d_h(j) = k$ is immediate.
    \item a chore transferred to $i$ from an agent $i^* \in N_{r^*}$ with $r^* < r < r'$. Then from Corollary~\ref{cor:Gain_Or_Lose} agents in $N_{r^*}$ only lose chores and the only chores of agents in $N_{r^*}$ are those allocated to them in $\s{MakeInitGroups}$. Property (iii) of Section~\ref{sec:overviewsimple} gives $d_h(j) = k$. \qedhere
\end{itemize}
\end{proof}

We may now show that in fact a stronger version of Invariant~\ref{inv:Group_Balance} holds. 
\begin{invariant}\label{inv:Group_Full_Balance} \normalfont
All groups maintain \textit{fully} balanced chores.
\end{invariant}
From Corollary~\ref{cor:Gain_Or_Lose} we know that a group cannot both gain and lose chores. We then consider three cases for a group $N_r$:
\begin{itemize}
    \item Suppose $N_r$ neither gains nor loses chores. Then agents in $N_r$ have only 1-chores from $\s{MakeInitGroups}$. Invariant~\ref{inv:Group_Balance} then guarantees Invariant~\ref{inv:Group_Full_Balance}.
    \item Suppose $N_r$ only loses chores. Then agents in $N_r$ again have only 1-chores and Invariant~\ref{inv:Group_Balance} gives Invariant~\ref{inv:Group_Full_Balance}.
    \item Suppose $N_r$ only gains chores. $\s{MakeInitGroups}$ initially guarantees balanced 1-chores in $N_r$. By Observation~\ref{lem:Gain_k} no agent in $N_r$ can gain a 1-chore, so they must remain balanced. Note that it is initially vacuously true that $N_r$ has balanced $k$-chores, so by Invariant~\ref{inv:Group_Balance} $N_r$ initially has fully balanced chores. We show that $N_r$ maintains fully balanced chores after gaining a chore by induction. Suppose $N_r$ has fully balanced chores before $i^* \in N_r$ gains a chore. From Lemma~\ref{lem:Gain_k} we know that the chore $j$ to be gained by $i^*$ is a $k$-chore. $N_r$ can only fail to have balanced $k$-chores after $i^*$ gains her chore if there exists $i \in N_r$ such that $|\x_{i_k}| < |\x_{{i^*}_k}|$ before $i^*$ gains $j$. Specifically, since $N_r$ has fully balanced chores, $|\x_{i_k}| = |\x_{{i^*}_k}| - 1$. We show that this cannot be. By the tiebreaking rules of Algorithm~\ref{alg:BalancedChores}, it must be that $|\x_{i^*}| < |\x_i|$, or $i$ would be selected over $i^*$ to gain a chore. This condition is equivalent to $|\x_{{i^*}_k}| + |\x_{{i^*}_1}| < |\x_{i_k}| + |\x_{i_1}|$. If $|\x_{i_k}| = |\x_{{i^*}_k}| - 1$, this implies that $|\x_{{i^*}_1}| < |\x_{i_1}| - 1$. Since $N_r$ maintains balanced 1-chores this is impossible. So, $N_r$ must maintain balanced $k$-chores and thus fully balanced chores.
\end{itemize}
As these cases are exhaustive, we have Invariant~\ref{inv:Group_Full_Balance}. We can now show Lemma~\ref{lem:BC_Term}.

\BCTerm*
\begin{proof}
Recall that $\s{MakeInitGroups}$ returns at most $R = N$ groups. Then by Condition (1) of Lemma~\ref{lem:Payment_Raises} there can be at most $n$ payment raises. From Observation~\ref{obs:High_To_Low}, each time a chore is transferred it moves to a lower group. As there are $m$ chores, there can be at most $mn$ chore transfers. Thus it is clear that Algorithm~\ref{alg:BalancedChores} terminates in $\s{poly}(n, m)$ time. It is then clear from the loop condition in Line 7 that Algorithm~\ref{alg:BalancedChores} returns a balanced allocation. 

We now show that $|\x|$ remains fPO. The initial allocation of $L$ returned by $\s{MakeInitGroups}$ is cost-minimizing, and thus guaranteed to be fPO. It is a competitive equilibrium such that all agents $i \in N$ have MPB ratio $\alpha_i = 1$. This can be seen from the fact that each agent has only 1-chores and every chore has payment 1. It is clear that Lines 2-4 maintain a competitive equilibrium, as for all $i \in N$, $j \in K$, we have $\alpha_{ij} = d_i(j)/p_j = 1$. So Algorithm~\ref{alg:BalancedChores} maintains a CE up to Line 7. We then show that Algorithm~\ref{alg:BalancedChores} returns a CE by induction. Suppose that we have a CE before a chore transfer or payment raise. It is clear that a chore transfer maintains a CE, as it is guaranteed that the transferred chore is MPB for the agent who receives it. We argue that a payment raise also maintains a CE. Consider a payment raise of a group $N_r$. By Lemma~\ref{lem:Payment_Raises} all agents outside of $N_r$ maintain the same MPB ratio, so they must continue to have MPB chores. Agents in $N_r$ also clearly continue to have MPB chores after raising payments of their own chores. Thus, a CE is maintained after any chore transfer or payment raise, so Algorithm~\ref{alg:BalancedChores} returns an fPO allocation.

We finally show that Algorithm~\ref{alg:BalancedChores} returns an allocation which is EF1. Suppose agent $i \in N_r$ has $s \geq 0$ 1-chores and $t \geq 0$ $k$-chores. Note that if $t = 0$, then $i$ cannot EF1-envy anyone. In such case, since we have seen Algorithm~\ref{alg:BalancedChores} returns a balanced allocation, for all $h \in N$ we would have $d_{i_{-1}}(\x_i) = s - 1 \leq |\x_h| \leq d_i(\x_h)$. Thus, we may assume that $i$ has a $k$-chore. Note that all of $i$'s $k$-chores are either:
\begin{itemize}
    \item a $K$-chore or
    \item a chore allocated by $\s{MakeInitGroups}$ to an agent in a group above $i$. This follows from Observation~\ref{obs:High_To_Low}.
\end{itemize}

We show that $i$ cannot EF1-envy any agent above her. Consider $i^* \in N_{r^*}$ with $r^* < r$. Then all of $i^*$'s chores are either a $K$-chore or a chore allocated by $\s{MakeInitGroups}$ to $i^*$ or an agent in or above $N_{r^*}$. Then by Property (iii) of Section~\ref{sec:overviewsimple}, $i$ must have disutility $k$ for all of $i^*$'s chores. Then since Algorithm~\ref{alg:BalancedChores} returns a balanced allocation, we have $d_{i_{-1}}(\x_i) = s + (t-1)k < (s + t - 1)k \leq d_i(\x_{i^*})$ and $i$ does not EF1-envy $i^*$.

We next show that $i$ cannot EF1-envy any agent $h$ in her same group. By Observation~\ref{lem:Gain_k} a $k$-chore of $h$ is also disutility $k$ for $i$. From Invariant~\ref{inv:Group_Full_Balance} we know that $N_r$ has fully balanced chores, so $h$ either:
\begin{itemize}
    \item has at least $s$ 1-chores and $(t-1)$ $k$-chores, or
    \item has at least $(s-1)$ 1-chores and $t$ $k$-chores.
\end{itemize}
Then, we have that $d_{i_{-1}}(\x_i) = s + (t-1)k \leq d_i(\x_h)$.

Finally, we show that $i$ cannot EF1-envy any agent $i'$ in a group below her. We claim that every $k$-chore of $i'$ is also disutility $k$ for $i$. If it is a $K$-chore, the result is immediate. Otherwise, it must have been transferred to $i'$ from a group above $N_r$. 

Indeed, no agent in $N_r$ can have lost a chore since $i$ has gained a chore. Suppose an agent $i''$ in a group below $i$ loses her first chore at time $T$ when she had $q$ chores. By the tiebreaker rules, it must be that $i$ had at most $q-1$ chores at time $T$, lest $i$ would be chosen over $i''$. In fact, by Invariant~\ref{inv:Almost_Fewer_Chores}, it must be that $i$ had exactly $q-1$ chores at $T$. We know that $i$ has gained a chore at some point in Algorithm~\ref{alg:BalancedChores}, and by Corollary~\ref{cor:Gain_Or_Lose} has not lost any chores. It cannot be that $i$ gained a chore before $T$, as this would imply she had only $q-2$ chores while $i''$ had $q$ chores, violating Invariant~\ref{inv:Almost_Fewer_Chores}. Yet, it also cannot be that she gained a chore at some time $T'$ after $T$. In this case, by Observation~\ref{obs:NondecreasingMinChores} it must be that all agents have at least $q-1$ chores before $i$ gains her chore. Since at time $T$ $i''$ lost a chore while she had $q$ chores, by Observation~\ref{obs:NonincreasingMaxChores}, at time $T'$ after $T$ it must be that all agents have at most $q$ chores. Thus, it must be that the allocation is already balanced at $T'$ and no transfer can occur. We conclude then that $i''$ cannot have lost a chore.

Then, both $i$ and $i'$ must retain all of their 1-chores from $\s{MakeInitGroups}$. By Properties (i) and (ii) from Section~\ref{sec:overviewsimple}, $i'$ can have at most one more 1-chore than $i$.  Then, since Algorithm~\ref{alg:BalancedChores} returns a balanced allocation, by the tiebreaker rules $i'$ must have at least $(t-1)$ $k$-chores, since $i'$ must receive her $q$-th $k$-chore before $i$ receives her $q+1$-th $k$-chore. It follows then that $d_{i_{-1}}(\x_i) = s + (t-1)k \leq d_i(\x_{i'})$ and $i$ does not EF1-envy $i'$. 

Thus, we conclude that Algorithm~\ref{alg:BalancedChores} terminates in $\s{poly}(n,m)$ time with a balanced EF1+fPO allocation.
\end{proof}

We include a few additional tools here to be used later.

\begin{restatable}{lemma}{No_Transfers}\label{lem:No_Transfers}
If Algorithm~\ref{alg:BalancedChores} performs no chore transfers, the allocation $\x$ returned is cost-minimizing. 
\end{restatable}
\begin{proof}
We know that $\s{MakeInitGroups}$ returns a cost-minimizing allocation. Then, since $j \in K$ gives the same disutility to any agent, any allocation of $K$ remains cost-minimizing. By Lemma~\ref{lem:BC_Term} Algorithm~\ref{alg:BalancedChores} is guaranteed to terminate. If there are no transfers, then it must be that this cost-minimizing allocation is returned.
\end{proof}

\begin{restatable}{lemma}{Raised1Chores}\label{lem:Raised1Chores}
Raised agents only ever have 1-chores which are disutility $k$ for any agent $i$ in an unraised group. 
\end{restatable}
\begin{proof}
Recall that a raised group contains an agent chosen to lose a chore, so by Lemma~\ref{lem:Gain_Or_Lose} a raised group cannot have ever gained a chore. An agent $h$ in a raised group can have only 1-chores allocated to them by $\s{MakeInitGroups}$. By Condition (2) of Lemma~\ref{lem:Payment_Raises} all unraised groups are always below all raised groups. Then, it follows from Property (iii) of Section~\ref{sec:overviewsimple} that all of $h$'s chores are disutility $k$ for agent $i$ in an unraised group.
\end{proof}

\begin{restatable}{lemma}{PseudoKChores}\label{lem:PseudoKChore}
A $k$-chore for an agent $i$ in an unraised group is disutility $k$ for all agents in an unraised group.
\end{restatable}
\begin{proof}
A $k$-chore $j$ for $i$ must either be a $K$-chore or a chore transferred to $i$. If $j$ is not a $K$-chore we show that $j$ must have come from a raised group. By Observation~\ref{obs:High_To_Low} $j$ must from an agent in a group $N_r$ above $i$'s group. Since an agent in $N_r$ lost $j$, by Corollary~\ref{cor:Gain_Or_Lose} no agent in $N_r$ has gained a chore. Then all chores belonging to agents in $N_r$ were allocated by $\s{MakeInitGroups}$ and by Property (iii) of Section~\ref{sec:overviewsimple} are disutility $k$ for $i$. By Condition (4) of Lemma~\ref{lem:Payment_Raises} though, $i$ has MPB ratio 1, so for $j$ to be MPB for $i$ it must have payment $k$. Then, $N_r$ must be a raised group. It follows from Lemma~\ref{lem:Raised1Chores} that $j$ is disutility $k$ for all unraised agents.
\end{proof}

We also obtain another result.
\begin{restatable}{lemma}{BC_EF1_PO_2ary}\label{lem:BC_EF1_PO_2ary} For 2-ary instances where $\forall i \in N, k_i \geq m$, an EF1+PO allocation can be found in strongly-polynomial time.
\end{restatable}
\begin{proof}
As in bivalued instances, we may re-scale valuations so that $d_i(j) \in \{1, k_i\}$ for each $i\in N, j\in M$. We show that by treating all $k_i$ as a fixed $k > 1$, Algorithm~\ref{alg:BalancedChores} returns an EF1+PO allocation $\x$ for the 2-ary instance. Note that Algorithm~\ref{alg:BalancedChores} produces the same output for any value of $k$. It follows immediately from Lemma~\ref{lem:BC_Term} then that $\x$ is EF1 even when the agents have different $k_i$'s. We now show that $\x$ is PO for the 2-ary instance where agent $i$ in fact has disutility $k_i$ for each of its $k$-chores under the bivalued transformation. We suppose that there exists an allocation $\y$ which Pareto dominates $\x$. 

We first note that $\y$ cannot increase the total number of chores between agents in $R$. This implies that there exists $i \in R$ such that $|\x_i| < |\y_i|$. By Lemma~\ref{lem:Raised1Chores} $i$ has only 1-chores under $\x$, so it must be that $d_i(\x_i) = |\x_i| < |\y_i| \leq d_i(\y_i)$ and $\y$ would not dominate $\x$. 

In addition, if an unraised agent $u$ receives $q$ $k$-chores under $\x$, then $u$ must receive at most $q$ $k$-chores under $\y$. Else, since for all $i \in N$ we have $k_i \geq m$, it must be that $d_u(\x_u) \leq (m-1)1 + (q)k_u \leq (q+1)k_u \leq d_u(\y_u)$.

Yet, if every unraised agent $u$ receives same the number of high disutility chores $q_u$ under both $\x$ and $\y$, then $\y$ cannot dominate $\x$. Indeed, outside of these chores every chore is allocated to an agent who has disutility 1 for it, so it must be that $\sum_{i \in N} d_i(\x_i) = \sum_{i \in N} d_i(\y_i)$ in this case.

Finally, we argue that if $\y$ dominates $\x$ then no unraised agent $u$ can receive fewer high disutility chores under $\y$ than in $\x$. Let $Q$ denote the total number of $k$-chores of unraised agents under $\x$. We have seen that $\y$ cannot increase the total number of chores between agents in $R$, so by Lemma~\ref{lem:Raised1Chores} and Lemma~\ref{lem:PseudoKChore} it must be that $\y$ also allocates at least $Q$ $k$-chores to unraised agents. Then, if an unraised agent $u$ receives fewer $k$-chores under $\y$, it must be that another unraised agent $u'$ receives strictly more $k$-chores. We have seen in this case that $\y$ cannot dominate $\x$.

Thus, there can be no allocation $\y$ which dominates $\x$, so $\x$ is PO.
\end{proof}

\subsection{EFX-Envy in Algorithm~\ref{alg:BalancedChores}}\label{sec:EFXLemmas} 

In this section we provide a few insights regarding EFX-Envy in an allocation $\x$ with agent groups $\{N_r\}_{r \in [R]}$. 

\begin{restatable}{lemma}{EFX_Best_Off}\label{lem:EFX_Best_Off}
If $\x$ is balanced and $i$ has only 1-chores, then $i$ cannot EFX-envy anyone.
\end{restatable}
\begin{proof}
We have that $\forall h \in N$, $d_i(\x_i) - 1 = |\x_i| - 1 \leq |\x_h| \leq d_i(\x_h)$.
\end{proof}

\begin{restatable}{lemma}{Higher_Group}\label{lem:EFX_Higher_Group}
If $\x$ is balanced, then for $i \in N_r$, $h \in N_{r'}$ with $r < r'$, $h$ does not EFX-envy $i$.
\end{restatable}
\begin{proof}
From Lemma~\ref{lem:Higher_k}, $\forall j \in \x_i$ we have $d_h(j) = k$. Then $\max_{j \in \x_h} d_h(\x_h \setminus j) \leq (|\x_h| - 1)k \leq |\x_i| \cdot k = d_h(\x_i)$.
\end{proof}

\begin{restatable}{lemma}{BalancedEnvy}\label{lem:EFX_BalancedEnvy}
In a balanced allocation, if agent $i$ EFX-envies agent $h$ and $i$ has $t$ $k$-chores, then there are at most $t-1$ chores $j \in \x_h$ such that $d_i(j) = k$. 
\end{restatable}
\begin{proof}
Suppose that $i$ has exactly $s$ 1-chores. We will assume that $s > 0$, otherwise it is clear she will not EFX-envy $h$. Suppose that $h$ has at least $t$ chores which are disutility $k$ for $i$. As $h$ has at least $s + t - 1$ chores, $\max_{j \in \x_i} d_i(\x_i \setminus j) = (s-1) + tk \leq d_i(\x_h)$.
\end{proof}

\begin{restatable}{lemma}{EFX_Equal_Chores_Group}\label{lem:EFX_Equal_Chores_Group}
In a group with fully balanced chores, if $i, h \in N_r$ both have exactly equal total chores and at least one 1-chore, then EFX-envy can only exist between them when $k > 2$.
\end{restatable}
\begin{proof}
From Lemma~\ref{lem:EFX_Best_Off} it must be that an agent has a $k$-chore to EFX-envy another agent. From Lemma~\ref{lem:PseudoKChore} a $k$-chore for $i$ is a $k$-chore for $h$ and vice versa. Without loss of generality, from Lemma~\ref{lem:EFX_BalancedEnvy} it must be that $|\x_{i_k}| = |\x_{h_k}| + 1$ for EFX-envy to exist. Then, since $i$ and $h$ have exactly equal total chores, it must also be that $|\x_{i_1}| + 1 = |\x_{h_1}|$. If $k \leq 2$, we have $\max_{j \in \x_i} d_i(\x_i \setminus j) = |\x_{i_1}| - 1 + |\x_{i_k}| \cdot k = |\x_{h_1}| - 2 + (|\x_{h_k}| + 1)k < |\x_{h_1}| + |\x_{h_k}| \cdot k \leq d_i(\x_h)$. So it must be that $k > 2$ for EFX-envy to exist.
\end{proof}

\subsection{Algorithm~\ref{alg:R1}: One Agent Group}

We first define two functions which appear in our algorithms.

\begin{itemize}
    \item $\s{Transfer}(j, i)$ takes chore $j$ from its current owner and gives $j$ to agent $i$. 
    \item $\s{Swap}(j_1, j_2)$ exchanges $j_1$ and $j_2$ between their respective owners.
\end{itemize}

\begin{algorithm}[!htbp]
\caption{Reducing EFX-Envy}\label{alg:R1}
\textbf{Input:} EF1 competitive equilibrium $(\x, \p)$\\
\textbf{Output:} EFX competitive equilibrium $(\x, \p)$
\begin{algorithmic}[1]
\If{$\x$ is not EFX}
\If{$|K| \equiv 2\ (\textrm{mod}\ 3)$}
	\State $i_1 \gets \s{argmin}_{i \in N} |\x_{i_k}|$
	\State $i_2 \gets i \in N \setminus \{i_1\}$
	\State $i_3 \gets i \in N \setminus \{i_1, i_2\}$
	\If{$\x_{i_1} \subseteq \s{MPB}_{i_2} \cap \s{MPB}_{i_3}$}
		\State $k_{i_2} \gets j \in \x_{i_2} \cap K$
		\If{$|\x_{i_2}| > |\x_{i_1}|$}
			\State $\s{Transfer}(k_{i_2}, i_1)$
		\Else
			\State $j_{i_1} \gets \{j \in \x_{i_1} \mid d_{i_2}(j) = 1\}$ 
			\State $\s{Swap}(k_{i_2}, j_{i_1})$
		\EndIf 
	\EndIf
	\If{$\x_{i_2} \subseteq \s{MPB}_{i_1} \cap \s{MPB}_{i_3}$}
		\State $k_{i_3} \gets j \in \x_{i_3} \cap K$
		\If{$|\x_{i_3}| > |\x_{i_2}|$}
			\State $\s{Transfer}(k_{i_3}, i_2)$
		\Else		
			\State $j_{i_2} \gets \{j \in \x_{i_2} \mid d_{i_3}(j) = 1\}$
			\State $\s{Swap}(k_{i_3}, j_{i_2})$
		\EndIf 
	\EndIf
	\If{$\x$ is not EFX}
	    \State $(\x, \p) \gets \s{Algorithm~\ref{alg:R1-2extra}}(\x, \p)$
	\EndIf
\ElsIf{$|K| \equiv 1\ (\textrm{mod}\ 3)$}	
	\State $i_1 \gets \s{argmax}_{i \in N} |\x_{i_k}|$
	\State $i_2 \gets i \in N \setminus \{i_1\}$
	\State $i_3 \gets i \in N \setminus \{i_1, i_2\}$
	\If{$\x_{i_2} \cup \x_{i_3} \subseteq \mpb_{i_1}$}
		\State $k_{i_1} \gets j \in \x_{i_1} \cap K$
		\If{$|\x_{i_1}| > |\x_{i_2}|$}
			\State $\s{Transfer}(k_{i_2}, i_1)$
		\Else
			\State $j_{i_2} \gets \{j \in \x_{i_2} \mid d_{i_1}(j) = 1 \}$
			\State $\s{Swap}(k_{i_1}, j_{i_2})$
		\EndIf 
	\EndIf
	\If{$\x_{i_1} \cup \x_{i_3} \subseteq \mpb_{i_2}$}
		\State $k_{i_2} \gets j \in \x_{i_2} \cap K$
		\If{$|\x_{i_2}| > |\x_{i_3}|$}
			\State $\s{Transfer}(k_{i_2}, i_3)$
		\Else
			\State $j_{i_3} \gets \{j \in \x_{i_3} \mid d_{i_2}(j) = 1\}$
			\State $\s{Swap}(k_{i_2}, j_{i_3})$
		\EndIf
	\EndIf
	\If{$\x$ is not EFX}
	    \State $(\x, \p) \gets \s{Algorithm~\ref{alg:R1-1extra}}(\x, \p)$
	\EndIf
\EndIf
\EndIf	
\State \Return $(\x, \p)$
\end{algorithmic}
\end{algorithm}

We now show that Algorithm~\ref{alg:R1} is only called when its input CE has specific properties.

\begin{restatable}{lemma}{R1_Input}\label{lem:R1_Input} An input $(\x, \p)$ for Algorithm~\ref{alg:R1} is always such that allocation $\x$ is cost-minimizing and has fully balanced chores.
\end{restatable}
\begin{proof}
When Algorithm~\ref{alg:R1} is called by Algorithm~\ref{alg:efxpo}, it must be that Algorithm~\ref{alg:BalancedChores} returned exactly one agent group. Then by Invariant~\ref{inv:Group_Full_Balance} this group has fully balanced chores. By Lemma~\ref{lem:Gain_Or_Lose} agents within a group cannot transfer chores amongst themselves, so Algorithm~\ref{alg:BalancedChores} performed no chore transfers and by Lemma~\ref{lem:No_Transfers} must return a cost-minimizing allocation.

When Algorithm~\ref{alg:R1} is called by Algorithm~\ref{alg:R2-1}, it must be that Algorithm~\ref{alg:BalancedChores} returned two agent groups with $N_1 = \{a\}$ and $N_2 = \{b, c\}$ and allocation $\x$ such that $|\x_{a_1}| = |\x_{c_1}|$. Note that by construction we have $|\x_{a_1}| \geq |\x_{b_1}| \geq |\x_{c_1}|$. Combining these, it must in fact be that $|\x_{a_1}| = |\x_{b_1}| = |\x_{c_1}|$. By Lemma~\ref{lem:BC_Term} Algorithm~\ref{alg:BalancedChores} returns a balanced allocation. Then, since the 1-chores are exactly equal, it must be that the $k$-chores are balanced. Thus, $\x$ has fully balanced chores. We show that it must also be that Algorithm~\ref{alg:BalancedChores} performed no chore transfers. If a transfer occurs, by Observation~\ref{obs:High_To_Low} it must be that $a \in N_1$ lost a chore and an agent in $N_2$ gained a chore. From Corollary~\ref{cor:Gain_Or_Lose} then $a$ cannot gain a chore and agents in $N_2$ cannot lose a chore. Recall from Lemma~\ref{lem:Gain_k} that agents only gain $k$-chores. Then both $b$ and $c$ must have at least $|\x_{b_1}| = |\x_{c_1}| = q$ chores before any chore transfer. Then $a$ cannot have lost a chore during the last transfer performed when she had $q+1$ chores. Thus, there can be no transfers and by Lemma~\ref{lem:No_Transfers} Algorithm~\ref{alg:BalancedChores} returns a cost-minimizing allocation.

Finally, when Algorithm~\ref{alg:R1} is called by Algorithm~\ref{alg:R2-2}, it must be that Algorithm~\ref{alg:BalancedChores} returned two agent groups with $N_1 = \{a, b\}$ and $N_2 = \{c\}$ and allocation $\x$ such that $|\x_{c_1}| \geq |\x_{b_1}|$ and $\x$ is not EFX. We first show that no transfers can have occurred. If if has, by Observation~\ref{obs:High_To_Low} it must be that an agent in $N_1$ lost a chore and $c$ gained a chore. Suppose a chore is transferred from $b$ to $c$. Then no agent in $N_1$ can have gained a chore after $\s{MakeInitGroups}$, so agents in $N_1$ have only 1-chores. Then by Lemma~\ref{lem:EFX_Best_Off} agents in $N_1$ cannot EFX-envy anyone. By Lemma~\ref{lem:EFX_Higher_Group} it cannot be $c$ EFX-envies $a$ or $b$. Then, $\x$ must be EFX, a contradiction. We conclude that no transfer can have occurred. By Lemma~\ref{lem:No_Transfers} $\x$ is cost-minimizing.

We now show that $\x$ is fully balanced. There are two cases in which Algorithm~\ref{alg:R2-2} invokes Algorithm~\ref{alg:R1}.
\begin{itemize}
    \item In the first case, $\x$ returned by Algorithm~\ref{alg:BalancedChores} is not EFX, $|\x_c| > |\x_a|$, and $|\x_{c_k}| \leq |\x_{b_k}|$. Since $\x$ is not EFX, it must be that some agent EFX-envies another agent. By Lemma~\ref{lem:EFX_Higher_Group} $c$ cannot be the EFX-envious agent. By construction we have that $|\x_{a_1}| \geq |\x_{b_1}|$ and we may assume that $b$ receives a $K$-chore before $a$ so $|\x_{b_k}| \geq |\x_{a_k}|$. By Lemma~\ref{lem:EFX_BalancedEnvy} and Lemma~\ref{lem:PseudoKChore} it must be that $b$ is the EFX-envious agent and $|\x_{b_k}| = |\x_{a_k}| + 1$, since Invariant~\ref{inv:Group_Full_Balance} guarantees $a$ and $b$ have balanced $k$-chores. From Properties (i) and (ii) of Section~\ref{sec:overviewsimple} we have that $|\x_{c_1}| \leq |\x_{b_1}| + 1 \leq |\x_{a_1}| + 1$. Note that by the tiebreaking rules $c$ is favored to receive a $K$-chore over $b$ if $|\x_b| \geq |\x_c|$. Thus, we have $|x_{c_k}| \geq |\x_{b_k}| - 1$. Combined with our initial conditions this gives $|\x_{b_k}| - 1 \leq |x_{c_k}| \leq |\x_{b_k}|$. Thus, $c$ must have the same number of $k$-chores as either $a$ or $b$, and $\x$ must have balanced $k$-chores. To see that $\x$ has balanced 1-chores, note that $|\x_{c_1}| \leq |\x_{b_1}| + 1 \leq |\x_{a_1}| + 1$. If $\x$ has unbalanced 1-chores, it can only be that $c$ has at least two fewer 1-chores than $a$. But then by tiebreaking rules and the fact that $|\x_c| > |\x_a|$, it must be that $c$ in turn has at least two more $k$-chores than $a$, which we have seen is impossible. Thus, $\x$ must have balanced 1-chores and thus be fully balanced in this case.
    \item In the second case, $\x$ is not EFX, $|\x_c| \leq |\x_a|$, and $|\x_{b_1}| > |\x_{c_1}|$. Again we have that $b$ must be EFX-envious and $|\x_{b_k}| = |\x_{a_k}| + 1$. Since $b$ is given a $K$-chore, by the tiebreaking rules it must be that $|\x_c| \geq |\x_b|$ Since $a$ and $b$ have fully balanced chores and $|\x_{b_k}| = |\x_{a_k}| + 1$, it must be that $|\x_b| \geq |\x_a|$. The given conditions include $|\x_a| \geq |\x_c|$, so combining these give $|\x_a| = |\x_b| = |\x_c|$. By construction and group balance we have $|\x_{b_1}| \leq |\x_{a_1}| \leq |\x_{b_1}| + 1$. From Properties (i) and (ii) of Section~\ref{sec:overviewsimple}. $|\x_{a_1}| \geq |\x_{c_1}|$. Then, $|\x_{c_1}| \geq |\x_{b_1}|$ implies $\x$ has balanced 1-chores. Then, since $|\x_a| = |\x_b| = |\x_c|$, it must be that the $k$-chores are balanced as well, so $\x$ is fully balanced. \qedhere
\end{itemize}
\end{proof}

Using Lemma~\ref{lem:R1_Input}, we then show that inputs for Algorithms \ref{alg:R1-2extra} and \ref{alg:R1-1extra} must also have certain properties.

\begin{restatable}{lemma}{R1_2extra_Input}\label{lem:R1_2extra_Input} An input $(\x, \p)$ for Algorithm~\ref{alg:R1-2extra} is always such that
\begin{enumerate}[label=(\roman*)]
    \item allocation $\x$ is cost-minimizing and has fully balanced chores,
    \item $|K| \equiv 2\ (\textrm{mod}\ 3)$, and
    \item there exists exactly one EFX-envy relationship.
\end{enumerate}
\end{restatable}
\begin{proof}
Lemma~\ref{lem:R1_Input} shows that input for Algorithm~\ref{alg:R1} satisfies (i). Line 2 of Algorithm~\ref{alg:R1} guarantees (ii). We show that Lines 3-20 of Algorithm~\ref{alg:R1} maintain (i) while giving (iii). Since the input $\x$ is fully balanced, if $|K| \equiv 2\ (\textrm{mod}\ 3)$ then there exists a unique agent $i_1 = \s{argmin}_{i \in N} |\x_{i_k}|$ and two agents $i_2$ and $i_3$ who we say have an \textit{extra} $k$-chore. In fact, since $\x$ is cost-minimizing it must be that for all $i \in N$ we have $|\x_{i_k}| \subseteq K$, so $i_2$ and $i_3$ have an extra $K$-chore. By Lemma~\ref{lem:EFX_BalancedEnvy} the only possible EFX-envy relationships are that $i_2$ and $i_3$ EFX-envy $i_1$, since they have more $K$-chores. However, if $i_1$ has any 1-chore which gives disutility $k$ to $i_2$ or $i_3$, then that agent no longer EFX-envies $i_1$ and we are done. Else, all of $i_1$'s 1-chores are MPB for both $i_2$ and $i_3$. Then, we aim to change which agents have the extra $K$-chores while maintaining cost-minimization and fully balanced chores. We aim to move one of the extra $K$-chores off of $i_2$ onto $i_1$. If $|\x_{i_2}| > |\x_{i_1}|$, we can simply transfer the extra $K$-chore from $i_2$ to $i_1$. Since $i_1$ had fewer total chores, the total chores remain balanced, and the same reasoning applies to the $k$-chores. Since no 1-chore is moved, they clearly remain balanced as well. Otherwise, it must be that $|\x_{i_2}| = |\x_{i_1}|$. It cannot be that $|\x_{i_2}| < |\x_{i_1}|$ since we have fully balanced chores and $|\x_{{i_2}_k}| > |\x_{{i_i}_k}|$. In this case, we swap the extra $K$-chore of $i_2$ with a 1-chore of $i_1$, which $i_1$ must have in order to be EFX-envied by $i_2$, and which we know must be MPB for $i_2$ as well. Now $i_1$ and $i_2$ have only swapped their number of 1-chores and $k$-chores, so $i_2$ has as many 1-chores and $k$-chores as $i_1$ did previously, and vice versa. It is clear then that we maintain fully balanced chores. Now $i_1$ and $i_3$ have the extra $K$-chores, and it can only be that these agents EFX-envy $i_2$. We perform the same type of MPB check. If $i_2$ has a 1-chore which is not MPB for either $i_1$ or $i_3$, we are done. Else, we repeat the same chore movement. Suppose another movement is needed. We claim that now there is at most one EFX-envy relationship. Indeed, if after each movement both agents with extra $K$-chores always EFX-envied the agent with fewer $K$-chores, then it must be that $a$, $b$, and $c$ all have disutility 1 for each 1-chore, showing that they in fact have identical valuations, which we have stated in Section~\ref{sec:overview3agents} to be impossible. Thus, there must be a chore $j$ belonging to the agent with fewer $K$-chores which is disutility $k$ for some agent with an extra $K$-chore, so at most one EFX-envy relationship can exist. Finally, Line 20 guarantees that if Algorithm~\ref{alg:R1-2extra} is called then that one possible EFX-envy relationship does exist. 
\end{proof}

\begin{restatable}{lemma}{R1_1extra_Input}\label{lem:R1_1extra_Input} An input $(\x, \p)$ for Algorithm~\ref{alg:R1-1extra} is always such that
\begin{enumerate}[label=(\roman*)]
    \item allocation $\x$ is cost-minimizing and has fully balanced chores,
    \item $|K| \equiv 1\ (\textrm{mod}\ 3)$, and
    \item there exists exactly one EFX-envy relationship.
\end{enumerate}
\end{restatable}
\begin{proof}
We proceed similarly to the proof of Lemma~\ref{lem:R1_2extra_Input}. Lemma~\ref{lem:R1_Input} shows that input for Algorithm~\ref{alg:R1} satisfies (i), Line 22 guarantees (ii), and we show that Lines 23-40 maintain (i) while obtaining (iii). It remains that an input allocation for Algorithm~\ref{alg:R1} is cost-minimizing and fully balanced. Now, when $|K| \equiv 1\ (\textrm{mod}\ 3)$, only one agent $i_1$ has an extra $K$-chore and two agents $i_2$ and $i_3$ have fewer $K$-chores. Our goal is to obtain an allocation in which the agent with the extra $K$-chore sees some 1-chore of another agent as disutility $k$, again while maintaining (i). The chore movements are equivalent to those described in Lemma~\ref{lem:R1_2extra_Input} and maintain (i). Also equivalently, if after two movements every agent has held the extra $K$-chore and seen every agent's 1-chores as MPB, then again we have a case of identical valuations which cannot be. Line 40 again guarantees the existence of exactly one EFX-envy relationship if Algorithm~\ref{alg:R1-1extra} is called, giving (iii).
\end{proof}

\subsection{Algorithm~\ref{alg:R1-2extra}: Two Extra $K$-chores}
\begin{algorithm}[!htbp]
\caption{Two extra $K$-chores}\label{alg:R1-2extra}
\textbf{Input:} EF1 competitive equilibrium $(\x, \p)$\\
\textbf{Output:} EFX competitive equilibrium $(\x, \p)$
\begin{algorithmic}[1]
\State $c \gets \s{argmin}_{i \in N} |\x_{i_k}|$
\State $b \gets i \in N \mid b \textrm{ EFX-envies } c$
\State $a \gets i \in N \mid a \textrm{ does not EFX-envy } c$
\State $j_c \gets j \in \x_c \setminus \mpb_a$
\If{$\exists j_b \in \x_b \setminus \mpb_c$}
	\State $k_b \gets j \in \x_b \cap K$
	\State $\s{Swap}(k_b, j_c)$
\ElsIf{$k \leq 2$}
	\State $j_b \gets \{j \in \x_b \mid d_c(j) = 1\}$
	\State $\s{Transfer}(j_b, c)$
\ElsIf{$|\x_{a_1}| < 2$}
	\While{$b$ EFX-envies $c$}
		\State $j_b \gets \{j \in \x_b \mid d_c(j) = 1\}$
		\State $\s{Transfer}(j_b, c)$
	\EndWhile
\ElsIf{$\exists j_a \in \x_a \cap \mpb_b \setminus \mpb_c$}
	\State $\s{Transfer}(j_a, b)$
	\State $k_b \gets j \in \x_b \cap K$
	\State $\s{Swap}(k_b, j_c)$
\ElsIf{$\exists j_a \in \x_a \cap \mpb_c \setminus \mpb_b$}
	\State $\s{Transfer}(j_a, c)$
	\If{$b$ EFX-envies $a$}
		\State $j_b \gets \{ j \in \x_b \mid d_c(j) = 1 \}$
		\State $\s{Transfer}(j_b, c)$
	\EndIf
\ElsIf{$\exists j_{a_1}, j_{a_2} \in \x_a \setminus \mpb_b \cup \mpb_c$}
	\If{$|\x_b| \leq |\x_a|$}
		\State $k_a \gets j \in \x_a \cap K$
		\State $\s{Transfer}(k_a, c)$
	\ElsIf{$\exists j \in \x_b \cup \x_c \text{ s.t. }  d_a(j) = 1$}
		\State $k_a \gets j \in \x_a \cap K$
		\State $\s{Transfer}(j, a)$
		\State $\s{Transfer}(k_a, c)$
	\Else
		\State $k_b \gets j \in \x_b \cap K$
		\State $\s{Transfer}(k_b, a)$
	\EndIf
\Else
	\While{$b$ EFX-envies $c$}
		\If{$|\x_b| \geq |\x_a|$}
			\State $j_b \gets \{j \in \x_b \mid d_c(j) = 1\}$
			\State $\s{Transfer}(j_b, c)$
		\Else
			\State $j_a \gets \{j \in \x_a \mid d_c(j) = 1\}$
			\State $\s{Transfer}(j_a, c)$
		\EndIf
	\EndWhile
\EndIf
\State \Return $(\x, \p)$
\end{algorithmic}
\end{algorithm}

We show that:

\begin{restatable}{lemma}{R1_2extra_EFX}\label{lem:R1_2extra_EFX} Algorithm~\ref{alg:R1-2extra} returns an EFX competitive equilibrium.
\end{restatable}
\begin{proof}
Algorithm~\ref{alg:R1-2extra} begins by defining agents $a$, $b$, and $c$ so that $c$ is the agent without an extra $k$-chore and $b$ is the unique agent who EFX-envies her. It follows that for all $j \in \x_{c_1}$, $d_b(j) = 1$. From Lemma~\ref{lem:R1_2extra_Input}, $\x$ is both cost-minimizing and fully balanced. In addition, $\exists j_c \in \x_c \setminus \mpb_a$. That is, $c$ must have some 1-chore $j_c$ which $a$ has disutility $k$ for.

In Line 5, we check if $b$ has a 1-chore $j_b$ which is disutility $k$ for $c$. Note that since $\x$ is cost-minimizing, $\x_{b_k} \subseteq K \subseteq \mpb_c$. If $b$ has such a chore, we swap a $k$-chore of $b$ with $j_c$. Now $b$ is the agent without the extra $k$-chore, but $b$ cannot be EFX-envied by $a$ as $b$ now has $j_c$ and similarly cannot be EFX-envied by $c$ due to having $j_b$. Thus, the allocation is EFX.

Then, in Line 8, it must now also be that $b$ does not have a 1-chore which is disutility $k$ for $c$, so $\x_b \subseteq \mpb_c$ and $\x_c \subseteq \mpb_b$. Additionally, since $k \leq 2$, we must have that $|\x_b| > |\x_c|$, as otherwise $b$ could not EFX-envy $c$ by Lemma~\ref{lem:EFX_Equal_Chores_Group}. Furthermore, since $k \leq 2$, $b$ EFX-envies $c$ by at most $2$ and transferring a 1-chore from $b$ to $c$ immediately removes $b$'s EFX-envy for $c$. Since $|\x_b| > |\x_c|$ before the transfer, the total number of chores remains balanced between all agents. Since each agent still has the same number of $k$-chores, no new EFX-envy can have been created, so the allocation is now EFX.

In Line 11, we check that $a$ has at least two 1-chores. If she does not, then we may transfer 1-chores from $b$ to $c$ until $b$ no longer EFX-envies $c$. Since $a$ has at most one 1-chore, we know that $b$ initially has at most two 1-chores. Note that $b$ will surely stop EFX-envying $c$ by the time she has transferred away all of her 1-chores. Yet, it must remain that $a$ and $b$ end with balanced chores and it cannot be that $a$ has grown to EFX-envy $b$. It also cannot be that $c$ begins to EFX-envy $b$, as she only incrementally receives 1-chores when $b$ EFX-envies her. Additionally, after the first transfer $a$'s bundle can be no better than $b$'s bundle for $c$, so $c$ also cannot EFX-envy $a$ and the allocation is EFX.

Now, in Line 15, we suppose that $a$ has a chore $j_a$ which is disutility 1 for $b$ but $k$ for $c$. We transfer $j_a$ to $b$ and swap a $k$-chore of $b$ with the chore $j_c \in \x_c$ which we know is disutility 1 for $b$ but $k$ for $a$. As a result, $a$ has lost a 1-chore, $b$ has lost a $k$-chore but gained two 1-chores, and $c$ has lost a 1-chore but gained a $k$-chore. We now check the EFX-envy relationships. Note that before the transfers $|\x_a| \geq |\x_c|$. So after the transfers $a$ and $c$ still have balanced chores. Since they now have the same number of $k$-chores, they cannot EFX-envy each other. Furthermore, neither of them can EFX-envy $b$, as although $b$ lost a $k$-chore, it gained $j_c$ which $a$ values as $k$ and $j_a$ which $c$ values as $k$. Finally, $b$ cannot EFX-envy $a$ or $c$. After the transfers, $b$ is the agent without an extra $k$ chore. Since $b$ has only one more chore in total, we know that $b$ has at most three more 1-chores than either $a$ or $b$. As we know that $k > 2$ from previous conditions, after removing a 1-chore $b$'s two remaining 1-chores are not worse than $a$ and $c$'s extra $k$-chore, so $b$ does not EFX-envy them. The allocation is then EFX.

The case handled in Line 19 is similar to that of Line 16. We now suppose that $a$ has a chore $j_a$ which is disutility 1 for $c$ but $k$ for $b$ and transfer this chore to $c$. We check the EFX condition. Clearly it remains that $a$ does not EFX-envy $c$. Although $b$ lost a $k$-chore, the addition of $j_c$ prevents $a$ from EFX-envying $b$, so $a$ does not EFX-envy anybody. We now examine $c$. Since $|\x_a| \geq |\x_c|$ before the transfers, after the transfers it remains that $a$ and $c$ have balanced chores. Since they now also have the same number of $k$-chores, $c$ does not EFX-envy $a$. After the transfers, $b$ has at least as many total chores as $c$. Although $c$ has an extra $k$-chore, since $b$ has $j_a$, $c$ views $b$ as having equally many $k$-chores, so $c$ does not EFX-envy $b$. Finally, $b$ now has one fewer $k$-chore than $a$ or $c$ but can have up to three more 1-chores. Since $k > 2$, after the removal of a 1-chore $b$ prefers her extra 1-chores over the extra $k$-chore of either $a$ or $c$. So $b$ does not EFX-envy anyone and the allocation is EFX.

In Line 24, we have that $b$ and $c$ both have disutility 1 for each other's 1-chores, $k > 2$, $a$ has at least two 1-chores, and $b$ and $c$ have the same valuation for all of $a$'s 1-chores. We also suppose that $a$ has two chores $j_{a_1}$ and $j_{a_2}$ which $a$ values at 1 but $b$ and $c$ value at $k$. Note that since both the $k$-chores and 1-chores are balanced between agents, we have that $|\x_c| \leq |\x_b|$ and $|\x_c| \leq |\x_a|$. We then have three cases.
\begin{enumerate}
	\item In Line 25 we have $|\x_b| \leq |\x_a|$. We transfer a $k$-chore from $a$ to $c$. If it is the case that $|\x_a| > |\x_c|$, then this transfer maintains the balance of total chores. Now $b$ and $c$ have the same number of $k$-chores, so they cannot EFX-envy each other. Since $a$ has only lost a chore she continues to not EFX-envy $b$ and $c$. We only need check that $b$ and $c$ do not EFX-envy $a$. Since $|\x_c| \leq |\x_b \leq |\x_a|$ before the transfer, the total chores remain balanced after the transfer. Although $a$ lost a $k$-chore, since $b$ and $c$ view $j_{a_1}$ and $j_{a_2}$ as $k$, they still do not EFX-envy $a$. Otherwise, if $|\x_c| = |\x_b| = |\x_a|$, the same arguments will hold, except between $a$ and $c$. Now after the transfer $c$ has two more chores than $a$. Suppose $a$ has $t$ $k$-chores and $p$ 1-chores ($c$ then has $t + 1$ $k$-chores and at most $p + 1$ 1-chores). Since $d_c(j_{a_1}) = d_c(j_{a_2}) = k$ and $k > 2$, we have $d_c(\x_c) - 1 \leq (t+1)k + p < (t+2)k + (p-2) \leq d_c(\x_a)$. So $c$ still cannot EFX-envy $a$.
	\item In Line 28 we have that $|\x_b| > |\x_a|$ and either $b$ or $c$ has a chore $j$ which is value 1 for $a$. Then we transfer $j$ to $a$ and transfer a $k$-chore from $a$ to $c$. Note that this preserves the balance of total chores. As in (1), $b$ and $c$ now have the same number of $k$-chores and cannot EFX-envy each other, $a$ has the fewest $k$-chores so she cannot EFX-envy $b$ or $c$, and $b$ and $c$ cannot EFX-envy $a$ due to $j_{a_1}$ and $j_{a_2}$.
	\item In Line 32, if $|\x_b| > |\x_a|$ but all 1-chores of $b$ and $c$ are value $k$ for $a$, we may simply transfer a $k$-chore from $b$ to $a$. Since $|\x_b| > |\x_a|$ this maintains the balance of total chores. Now, $a$ has both of the extra $k$-chores, so $b$ and $c$ do not EFX-envy her, and since $b$ and $c$ have the same number of $k$-chores they do not EFX-envy each other. Finally, $a$ cannot EFX-envy $b$ or $c$ as she sees all of their chores as $k$.
\end{enumerate}

We now have the final case in Line 35, where it must be that $b$ and $c$ value all of each other's 1-chores as 1 and $a$ has at most one 1-chore which is value $k$ for $b$ and $c$. In this case, we transfer 1-chores from $a$ and $b$ to $c$ until $b$ no longer EFX-envies $c$, while maintaining a balanced number of chores between $a$ and $b$ so they do not begin to EFX-envy each other. If $b$ has at least as many chores as $a$, we transfer a 1-chore from $b$ to $c$. Otherwise, if $|\x_a| > |\x_b|$, we transfer a chore out of $a$ instead, so that $a$ and $b$ will maintain balanced total chores and thus not EFX-envy each other. It is clear that $a$ still does not EFX-envy $c$, and $b$ must stop EFX-envying $c$ by the time when she has given away all of her 1-chores. Then, since $c$ only receives 1-chores incrementally, she cannot grow to EFX-envy $b$. And, since $a$'s bundle is at least as bad as $b$'s bundle for $c$, $c$ also cannot EFX-envy $a$. Thus, the allocation is EFX.

Then, since all transfers in all cases are done only with MPB chores, it remains that we have a CE, and thus the allocation is fPO.
\end{proof}

\subsection{Algorithm~\ref{alg:R1-1extra}: One Extra $K$-chore}

\begin{algorithm}[!htb]
\caption{One extra $K$-chore}\label{alg:R1-1extra}
\textbf{Input:} EF1 competitive equilibrium $(\x, \p)$\\
\textbf{Output:} EFX competitive equilibrium $(\x, \p)$
\begin{algorithmic}[1]
\State $a \gets \s{argmax}_{i \in N} |\x_{i_k}|$
\State $b \gets i \in N \mid a \textrm{ EFX-envies } b$
\State $c \gets i \in N \mid a \textrm{ does not EFX-envy } c$
\State $j_c \gets j \in \x_c \setminus \mpb_a$
\State $k_a \gets j \in \x_a \cap K$
\State $j_b \gets \{j \in \x_b \mid d_a(j) = 1$\}
\If{$\exists j \in \x_c \setminus \mpb_b$}
	\State $\s{Swap}(k_a, j_b)$
	\While{$b$ EFX-envies $a$}
		\State $j_{b^*} \gets \{ j \in \x_b \mid d_a(j) = 1 \}$
		\State $\s{Transfer}(j_{b^*}, a)$ 
	\EndWhile
\ElsIf{$\exists j_{b'} \in \x_b \text{ s.t. } d_c(j_{b'}) = 1$}
	\State $\s{Swap}(j_{b'}, j_c)$
	\If{$\x$ is not EFX}
		\State $j_{c^*} \gets j \in \x_c \cap L$
		\State $\s{Swap}(k_a, j_{c^*})$
	\EndIf
	\While{$c$ EFX-envies $a$ or $b$}
		\State $j_{c^*} \gets j \in \x_c \cap L$
		\If{$|\x_a| < |\x_b|$}
			\State $\s{Transfer}(j_{c^*}, a)$
		\Else
			\State $\s{Transfer}(j_{c^*}, b)$
		\EndIf
    \EndWhile
\ElsIf{$|\x_b|_1 = 1$}
	\State $\s{Swap}(k_a, j_b)$
\Else
	\State $j_{c'} \gets \{ j \in \x_c \mid d_b(j) = 1\}$
	\State $\s{Transfer}(k_a, c)$
	\State $\s{Transfer}(j_{c'}, b)$
	\State $\s{Transfer}(j_b, a)$
\EndIf
\State \Return $(\x, \p)$
\end{algorithmic}
\end{algorithm}

We show that:

\begin{restatable}{lemma}{R1_1extra_EFX}\label{lem:R1_1extra_EFX} Algorithm~\ref{alg:R1-1extra} returns an EFX competitive equilibrium.
\end{restatable}
\begin{proof}
Algorithm~\ref{alg:R1-1extra} begins by defining agents $a$, $b$, and $c$ so that $a$ is the agent with the extra $k$-chore, and $a$ EFX-envies $b$ but not $c$. It follows that $a$ sees all of $b$'s 1-chores as 1 also, and there exists some 1-chore $j_c \in \x_c$ which $a$ values as $k$.

We first check in Line 7 if $c$ has a chore $j$ which is not MPB for $b$, meaning it has value 1 for $c$ but $k$ for $b$. If so, then we swap a $k$-chore of $a$ with a 1-chore of $b$. Only $b$ can EFX-envy another agent as $b$ now has the extra $k$-chore. It cannot be that $b$ EFX-envies $c$ due to $j$, so $b$ can only EFX-envy $a$. If this is the case, we simply transfer 1-chores from $b$ to $a$ until the EFX-envy disappears. We check that $c$ does not EFX-envy $a$ as they now have the same number of $k$-chores and $a$ has weakly more 1-chores than $c$. It also cannot be that $c$ EFX-envies $b$, as we have $d_c(\x_c) - 1 \leq d_a(\x_a) < d_a(\x_b) - 1 = d_b(\x_b) - 1 \leq d_c(\x_b) - 1 < d_c(\x_b)$. Thus, the allocation is EFX.

In Line 12, if $c$ has no 1-chore which is value $k$ for $b$, we check whether $b$ has a 1-chore $j_{b'}$ which is also value 1 for $c$. If so, we swap $j_{b'}$ with $j_c$. Now if the allocation is not EFX, it can only be that $a$ EFX-envies $c$ and values of $c$'s 1-chores as 1 also. Then, in fact, it must be that $c$ only has universal 1-chores which are value 1 for all agents. We now swap a $k$-chore from $a$ with a universal 1-chore from $c$, so that now only $c$ can EFX-envy someone. If $c$ does EFX-envy someone, then we transfer a 1-chore out of $c$ to whomever of $a$ and $b$ that has the fewest chores. Note, this is not necessarily the agent that $c$ EFX-envies. We show that this can be done until the allocation is EFX. As $a$ and $b$ now have equal $k$-chores, and we maintain the fact that they have balanced total chores, they cannot EFX-envy each other. We also have that $d_a(\x_a) - 1 \leq d_b(\x_b) - 1 \leq d_c(\x_b) - 1 \leq d_c(\x_c) \leq d_a(\x_c)$, so $a$ and $b$ cannot EFX-envy $c$. In turn, $c$ is guaranteed to stop EFX-envying both $a$ and $b$ before she gives away all of her 1-chores, so the allocation is EFX.

In Line 23 we do a simple check on the number of 1-chores of $b$. If $b$ has exactly one 1-chore, swapping this chore with the extra $k$-chore from $a$ is immediately EFX. 

Finally, in Line 25, it must be that $b$ values all of $c$'s 1-chores as 1, $c$ values all of $b$'s 1-chores as $k$, and $a$ values all of $b$'s 1-chores as 1. We can then rotate chores so that $a$ gives a $k$-chore to $c$, $c$ gives a 1-chore to $b$, and $b$ gives a 1-chore to $a$. Now $c$ has the extra $k$-chore and is the only agent who can EFX-envy. However, $c$ cannot EFX-envy $a$ as $a$ has received a chore from $b$, which $c$ values as $k$. In addition, $b$ retains at least one 1-chore which it had initially, which $c$ also values as $k$. Thus, $c$ cannot EFX-envy $b$ either, so the allocation is EFX.

Since all transferred chores are MPB for their recipients, we maintain a CE and thus the allocation returned is fPO.
\end{proof}

\subsection{Algorithm~\ref{alg:R2-1}: $\mathbf{R=2, |N_1| = 1}$}

\begin{algorithm}[!htb]
\caption{When $R=2$ and $|N_1| = 1$}\label{alg:R2-1}
\textbf{Input:} EF1 competitive equilibrium $(\x, \p)$\\
\textbf{Output:} EFX competitive equilibrium $(\x, \p)$
\begin{algorithmic}[1]
\If{$\x$ is not EFX}
	\State $a \gets i \in N_1$
	\State $b \gets \s{argmax}_{i \in N_2} |\x_{i_1}|$
	\State $c \gets i \in N_2 \setminus \{b\}$
	\If{$|\x_{a_1}| = |\x_{c_1}|$}
		\State $(\x, \p) \gets \s{Algorithm~\ref{alg:R1}(\x, \p)}$
	\ElsIf{$|\x_a| < |\x_c|$}
		\State $k_c \gets j \in \x_c \cap K$
		\State $\s{Transfer}(k_c, a)$
	\ElsIf{$|\x_b| < |\x_a|$ or $|\x_{b_1}| \geq 3$}
		\If{$\exists j_a \in \x_a \text{ s.t. } j_a \in \mpb_b$}
			\State $\s{Transfer}(j_a, b)$
		\Else
			\State raise payments of $\x_a$ by a factor of $k$
			\State $j_a \gets j \in \x_a$
			\State $\s{Transfer}(j_a, b)$
		\EndIf
	\Else
		\If{$\exists j_c \in \x_c \text{ s.t. } d_b(j_c) = 1$}
			\State $\s{Transfer}(j_c, b)$
		\Else
			\State $k_c \gets j \in \x_c \cap K$
			\State $j_b \gets 
			\{ j \in \x_b \mid d_c(j) = 1$ \}
			\State $\s{Swap}(k_c, j_b)$
		\EndIf
	\EndIf
\EndIf
\State \Return $(\x, \p)$
\end{algorithmic}
\end{algorithm}

We show that:

\begin{restatable}{lemma}{R2_1_EFX}\label{lem:R2_1_EFX} Algorithm~\ref{alg:R2-1} returns an EFX competitive equilibrium.
\end{restatable}
\begin{proof}
We first consider some properties of any input of Algorithm~\ref{alg:R2-1}. When $R = 2$ and $|N_1| = 1$, $a$ cannot EFX-envy $b$ or $c$, as they will always have at least as many $k$-chores as $a$, and $b$ and $c$ cannot EFX-envy $a$ as they see all of her chores as $k$. Furthermore, without loss of generality we may assume that $c$ receives her first $k$-chore before $b$ does, so $c$ will always have at least as many $k$-chores as $b$. Thus, the EFX-envy must exist from $c$ towards $b$. We now cover the conditions in Lines 5, 7, 10, and 17.

In Line 5, when $c$ has exactly as many 1-chores as $a$, we have that all agents have equal number 1-chores and balanced $k$-chores. Then, we may use Algorithm~\ref{alg:R1} to find an EFX allocation.

In Line 7, we now suppose that $c$ has fewer 1-chores than $a$, and $|\x_a| < |\x_c|$. Note that the latter condition only holds when no transfers occur in Algorithm~\ref{alg:BalancedChores}. In this case, we transfer a universal $k$-chore from $c$ to $a$. The number of chores remains balanced between all agents. Now $b$ and $c$ possess the same number of $k$-chores, so they cannot EFX-envy each other, and it is clear that they remain unenvious of $a$. We need only show that $a$ does not begin to EFX-envy either $b$ or $c$. Since $c$ has fewer 1-chores than $a$, yet $|\x_a| < |\x_c|$, it must be that $c$ has at least two more $k$-chores than $a$ and $b$ has at least one more $k$-chore than $a$. Then, even after $c$ loses a $k$-chore to $a$, $a$ has at most as many $k$-chores as $b$ and $c$, so she does not EFX-envy them. 

Now in Line 10, we first consider the case when $|\x_b| < |\x_a|$. We transfer a chore from $a$ to $b$, which both $b$ and $c$ value as $k$. The number of chores remains balanced between all agents, and $b$ and $c$ now have the same number of $k$-chores. Since they still see all of $a$'s chores as $k$, they remain unenvious of $a$, and $a$ herself clearly remains unenvious as she has only lost a chore. Thus, the allocation is EFX. If it is not the case that $|\x_b| < |\x_a|$, we have $|\x_c| \leq |\x_a| \leq |\x_b|$. Since $|\x_c| < |\x_b|$ implies that $b$ has as many $k$-chores as $c$, and thus no EFX-envy between them, we must in fact have $|\x_c| = |\x_a| = |\x_b| = p$. From $|\x_c| = |\x_b|$ and the fact that $c$ EFX-envies $b$ we also conclude that $b$ has one more 1-chore than $c$, so it must be that $k > 2$ for $c$ to EFX-envy $b$. Now suppose $b$ has at least three 1-chores and she receives a chore from $a$ with value $k$. Clearly $a$ remains unenvious of $b$ and $c$. By Lemma~\ref{lem:EFX_BalancedEnvy} $b$ and $c$ do not EFX-envy each other, as they have equal number $k$-chores and balanced total chores. In addition, $c$ cannot EFX-envy $a$ as the two still have balanced total chores and $c$ sees all of $a$'s chores as $k$. We need only check that $b$ does not begin to EFX-envy $a$. Since $b$ has at least three 1-chores and at most $p - 2$ $k$-chores, and we know $k > 2$, we have $d_b(\x_b) - 1 \leq (p-2)k + 2 \leq (p-1)k = d_b(\x_a)$. So $b$ does not EFX-envy $a$ and the allocation is EFX.

In Line 17, as none of the previous conditions hold, we must have $|\x_c| = |\x_a| = |\x_b| = p$, $|\x_b|_1 < 3$, and $k > 2$. We argue that in fact $|\x_b|_1 = 2$ and $|\x_c|_1 = 1$. That is, $b$ must have exactly two 1-chores and $c$ must have exactly one 1-chore. Note that for $c$ to have EFX-envy for $b$, $c$ must have at least one 1-chore. Then we must have $|\x_b|_1 \neq |\x_c|_1$, as otherwise $|\x_b| = |\x_c|$ would imply that $b$ and $c$ have the same number of $k$-chores and no EFX-envy between them. So, we must have $|\x_b|_1 > |\x_c|_1 \geq 1$ and in conjunction with $|\x_b|_1 < 3$ the desired result follows. Now suppose there exists a 1-chore of $c$'s which is also value 1 for $b$. Upon transferring this chore to $b$, we can see that $c$ no longer EFX-envies $b$, as $c$ now only has $p-1$ $k$-chores, so $d_c(\x_c) - k = (p - 2)k < (p - 2)k + 3 \leq d_c(\x_b)$. We also check that $b$ does not start EFX-envying $c$. Since $k > 2$ we have $d_b(\x_b) - 1 = (p - 2)k + 2 \leq (p - 1)k = d_b(\x_c)$. It is easy to verify that $a$ still does not EFX-envy $b$ or $c$, as $a$ still has balanced total chores with respect to either and fewer $k$-chores than either. Thus, the allocation is EFX.

Note that all transfers consist of only MPB chores, so it remains that we have a CE and the allocation is fPO.
\end{proof}

\subsection{Algorithm~\ref{alg:R2-2}: $\mathbf{R=2, |N_1| = 2}$}

\begin{algorithm}[!htb]
\caption{When $R=2$ and $|N_1| = 2$}\label{alg:R2-2}
\textbf{Input:} EF1 competitive equilibrium $(\x, \p)$\\
\textbf{Output:} EFX competitive equilibrium $(\x, \p)$
\begin{algorithmic}[1]
\If{$\x$ is not EFX}
	\State $a \gets \s{argmax}_{i \in N_1} |\x_{i_1}|$
	\State $b \gets i \in N_1 \setminus \{a\}$
	\State $c \gets i \in N_2$
	\If{$|\x_c| > |\x_a|$}
	    \If{$|\x_{c_k}| > |\x_{b_k}|$}
	        \State $k_c \gets j \in \x_c \cap K$
	        \State $\s{Transfer}(k_c, a)$
	    \Else
	        \State $(\x, \p) \gets \s{Algorithm~\ref{alg:R1}(\x, \p)}$
	    \EndIf
	\Else
	    \If{$|\x_{c_1}| \geq |\x_{b_1}|$}
	        \State $(\x, \p) \gets \s{Algorithm~\ref{alg:R1}(\x, \p)}$
	    \Else
		    \If{$\exists j_c \in \x_c \text{ s.t. } d_b(j_c) = 1$}
		    	\State $k_b \gets j \in \x_b \cap K$
		    	\State $\s{Swap}(j_c, k_b)$
	    	\Else
		    	\State $k_c \gets j \in \x_c \cap K$
		    	\State $j_a \gets \{j \in \x_a \mid d_b(j) = 1\}$
		    	\State $\s{Transfer}(k_c, a)$
		    	\State $\s{Transfer}(j_a, b)$
	    	\EndIf
	    \EndIf
	\EndIf
\EndIf
\State \Return $(\x, \p)$
\end{algorithmic}
\end{algorithm}

We show that:

\begin{restatable}{lemma}{R2_2_EFX}\label{lem:R2_2_EFX} Algorithm~\ref{alg:R2-2} returns an EFX competitive equilibrium.
\end{restatable}
\begin{proof}

We have shown in the proof of Lemma~\ref{lem:R1_Input} that when $R = 2$ and $|N_1| = 2$, if any chore transfers occur in the execution of Algorithm~\ref{alg:BalancedChores} the output of Algorithm~\ref{alg:BalancedChores} is already EFX. Thus, before Line 5 of Algorithm~\ref{alg:R2-2}, we must have a cost-minimizing CE where each agent has only 1-chores or $K$-chores. We define agent $a$ to be the agent with weakly greater 1-chores in $N_1$, $b$ to be the remaining agent in $N_1$, and $c$ to be the agent in $N_2$. Without loss of generality, we may assume that $b$ receives a $k$-chore before $a$ in Algorithm~\ref{alg:BalancedChores} so that $|\x_{b_k}| \geq |\x_{a_k}|$. Furthermore, if EFX-envy exists in $\x$, it must be that $|\x_{b_k}| = |\x_{a_k}| + 1$ and $b$ is the EFX-envious agent.

In Line 5, we check that $c$ has strictly more total chores than $a$. If so, we check in Line 6 that $c$ has strictly more $K$-chores than $b$. If both conditions are satisfied, we directly transfer a $K$-chore from $c$ to $a$. Since $|\x_{c_k}| > |\x_{b_k}| \geq |\x_{a_k}| + 1$, it remains after the transfer that $c$ has at least as many $K$-chores as $a$ and $b$, so they cannot EFX-envy $c$. It is also the case that $a$ and $b$ now have equal $K$-chores, so they cannot EFX-envy each other. It remains that $c$ cannot EFX-envy $a$ or $b$, so the allocation is EFX. Else, if $|\x_c| > |\x_a|$ but $|\x_{c_k}| \leq |\x_{b_k}|$, we call Algorithm~\ref{alg:R1}. Algorithm~\ref{alg:R1} either finds an EFX allocation immediately or calls Algorithm~\ref{alg:R1-2extra} or Algorithm~\ref{alg:R1-1extra}. We have seen from Lemma~\ref{lem:R1_2extra_EFX} and Lemma~\ref{lem:R1_1extra_EFX} that these algorithms return an EFX competitive equilibrium. The same argument applies to the execution of Line 13. 

We then check the conditions in Line 14. It must be that $|\x_c| \leq |\x_a|$ and $|\x_{c_1}| < |\x_{b_1}|$. Note that $b$ is given a $K$-chore. By the tiebreaker rules, it be then that $|\x_c| \geq |\x_b|$. In addition, Invariant~\ref{inv:Group_Full_Balance} guarantees that $a$ and $b$ have fully balanced chores. Since $|\x_{b_k}| = |\x_{a_k}| + 1$, we have that $|\x_b| \geq |\x_a|$. Combining these results we see that $|\x_c| \geq |\x_b| \geq |\x_a| \geq |\x_c|$, so in fact $|\x_c| = |\x_b| = |\x_a|$. It follows then that $|\x_{a_1}| = |\x_{b_1}| + 1 > |\x_{c_1} + 1$, i.e., $a$ has exactly one more 1-chore than $b$, and $b$ has at least one 1-chore than $c$. Similarly, we have $|\x_{c_k}| > |\x_{b_k}| = |\x_{a_k}| + 1$, i.e., $c$ has at least one more $K$-chore than $b$, and $b$ has exactly one more $K$-chore than $a$. In this case, we first check if $c$ has a chore $j_c$ which is disutility 1 for $b$. If so, we swap a $K$-chore of $b$ with $j_c$. Then $a$ and $b$ cannot EFX-envy each other as they have equal $K$-chores now, and they clearly continue to not EFX-envy $c$. At the same time $c$ can still not EFX-envy $a$ as $c$ has disutility $k$ for all of her chores. Regarding $b$, $c$ now views all but one of $b$'s chores as $k$. Then, since every agent has exactly $q$ chores, we have $\max_{j \in \x_c} d_c(\x_c \setminus j) \leq (q-1)k \leq (q-1)k + 1 = d_c(\x_b)$. So $c$ does not EFX-envy $b$ and the allocation is EFX.

Else, in Line 18 we still have that $|\x_c| = |\x_b| = |\x_a| = q$ but now $b$ must have disutility $k$ for all of $c$'s 1-chores. In this case, we transfer a $K$-chore from $c$ to $a$, and a 1-chore from $a$ to $b$. Notably, since $c$ cannot be EFX-envied by $b$ here, it must be that $b$ EFX-envies $a$ all of $a$'s chores are MPB for $b$. After this transfer $c$ continues to not envy anybody. Now $a$ and $b$ have balanced chores and equal $K$-chores, so they do not EFX-envy each other. It remains that $a$ and $c$ have balanced chores, and $c$ still has at least as many $K$-chores as $a$, so $a$ does not EFX-envy $c$. We now examine $b$ and $c$. It remains that $c$ has at least as many $K$-chores as $b$, but $b$ now has two more total chores than $c$. Thus, after removing a chore it can EFX-envy $c$ by at most 1. However, since we had to have had EFX-envy between $a$ and $b$ when $|\x_b| = |\x_a|$, by Lemma~\ref{lem:EFX_Equal_Chores_Group} it must be that $k > 2$ in this instance. Then, since $b$ has disutility $k > 2$ for all of $c$'s 1-chores, $b$ does not EFX-envy $c$. Note, it must be that $c$ has a 1-chore since it could not form a group by itself with no initial chores. Thus, the allocation is EFX.

As only MPB chores are transferred, it must be that we maintain a CE and the allocation returned is fPO.
\end{proof}

\subsection{Algorithm~\ref{alg:efxpo}: Obtaining EFX+fPO}\label{sec:efxpoFinal}
\begin{lemma}\label{lem:bivalued3agents}
Given a bivalued chore allocation instance $(N,M,D)$ with three agents, Algorithm~\ref{alg:efxpo} computes an EFX + fPO allocation in strongly polynomial-time.
\end{lemma}
\begin{proof}
We show that Algorithm~\ref{alg:efxpo} computes such an allocation in strongly-polynomial time.  Algorithm~\ref{alg:efxpo} first calls Algorithm~\ref{alg:BalancedChores}. Then,  \begin{itemize}
    \item If $R=1$, Algorithm~\ref{alg:R1} is invoked. Then, it must be that either Algorithm~\ref{alg:R1} finds an EFX+fPO allocation, or one of Algorithms \ref{alg:R1-2extra} and \ref{alg:R1-1extra} is called. From Lemmas \ref{lem:R1_2extra_EFX} and \ref{lem:R1_1extra_EFX} we know that an EFX+fPO allocation is found in either case.
    \item If $R=2$, then either Algorithm~\ref{alg:R2-1} or Algorithm~\ref{alg:R2-2} is invoked. From Lemmas \ref{lem:R2_1_EFX} and \ref{lem:R2_2_EFX} we know that both algorithms give an EFX+fPO allocation.
    \item Finally, if $R=3$, we show that the output of Algorithm~\ref{alg:BalancedChores} is already EFX+fPO. Suppose the groups are such that $N_1 = \{a\}$, $N_2 = \{a\}$, and $N_3 = \{c\}$. From Lemma~\ref{lem:BC_Term} we have that the allocation is balanced. Then from Lemma~\ref{lem:EFX_Higher_Group} it cannot be that $c$ EFX-envies $a$ or $b$, and it cannot be that $b$ EFX-envies $a$. If an agent belongs to a raised group, then by Lemma~\ref{lem:Raised1Chores} they have only 1-chores, and so by Lemma~\ref{lem:EFX_Best_Off} agents in such groups cannot EFX-envy anyone. We then need only consider an agent in an unraised group who may EFX-envy an agent in a group below her. By construction of the groups, we have after $\s{MakeInitGroups}$ is called that $|\x_{a_1}| \geq |\x_{b_1}| \geq |\x_{c_1}|$. By Condition (3) of Lemma~\ref{lem:Payment_Raises}, agents in unraised groups cannot have lost a chore. By Lemma~\ref{lem:Gain_k} they may only gain $k$-chores. Then, by the tiebreaker rules, we have that $|\x_{c_k}| \geq |\x_{b_k}| \geq |\x_{a_k}|$. Since any unraised agent sees every $k$-chore of every other unraised agent as $k$ by Lemma~\ref{lem:PseudoKChore} and the allocation is balanced, it cannot be that $a$ EFX-envies $b$ or $c$, or that $b$ EFX-envies $c$. Thus, the allocation is EFX, and we have seen from Lemma~\ref{lem:BC_Term} that it is also fPO.
\end{itemize}
Thus, Algorithm~\ref{alg:efxpo} finds an EFX+fPO allocation. We now show that Algorithm~\ref{alg:efxpo} terminates in strongly-polynomial time. By Lemma~\ref{lem:BC_Term} Algorithm~\ref{alg:BalancedChores} terminates in $\s{poly}(n,m)$ time. If Algorithm~\ref{alg:R1} is not called, it is clear that Algorithm~\ref{alg:R2-1} and Algorithm~\ref{alg:R2-2} both terminate in $\s{poly}(n,m)$ time as they perform only constantly many chore transfers. If Algorithm~\ref{alg:R1} is called, it performs only constantly many chore transfers before calling either Algorithm~\ref{alg:R1-2extra} or Algorithm~\ref{alg:R1-1extra}. Algorithm~\ref{alg:R1-2extra} performs only constantly many transfers outside of its two while loops in Line 12 and Line 36. In both of these cases, we show in Lemma~\ref{lem:R1_2extra_EFX} that the loop terminates after at most $m$ transfers. The same is done for Algorithm~\ref{alg:R1-1extra}, which performs constantly many tranfers outside of its two while loops in Lines 19 and 17. Lemma~\ref{lem:R1_1extra_EFX} shows that they too terminate in at most $m$ transfers. Thus, it must be that Algorithm~\ref{alg:efxpo} terminates in strongly-polynomial time.
\end{proof}

\section{Illustrative Examples}\label{app:examples}

\begin{example}The algorithm of \label{ex:efx-alg-not-po}\cite{zhou22efx} does not necessarily return a PO allocation. \end{example} 
Zhou and Wu~\cite{zhou22efx} use the \RR procedure as a subroutine in their algorithm. We show that \RR need not give a PO allocation. Consider the following instance with three agents $\{a,b,c\}$ and three chores $\{j_1,j_2,j_3\}$.

\begin{table}[ht]
\centering
\begin{tabular}{|c||c|c|c|}\hline
    & $j_1$ & $j_2$ & $j_3$ \\\hline 
$a$ & 1 & 1 & 1 \\\hline
$b$ & $k$ & 1 & $k$ \\\hline
$c$ & 1 & $k$ & $k$ \\\hline
\end{tabular}
\end{table}
Consider an execution of the \RR algorithm using the order $a,b,c$. First $a$ picks $j_1$, then $b$ picks $j_2$, then $c$ picks $j_3$. This gives a disutility vector of $v_1 = (1,1,k)$ for agents $a,b,c$. However allocating $j_3$ to $a$, $j_2$ to $b$ and $j_1$ to $c$ gives a utility vector of $v_2 = (1,1,1)$, which Pareto-dominates $v_1$. Hence the outcome of \RR is not necessarily PO, thus neither is outcome of the algorithm of \cite{zhou22efx}.

\begin{example}\label{ex:MakeInitGroups} Execution of Algorithm $\s{MakeInitGroups}$
\end{example}
Consider the chore allocation instance below (with only $L$-chores).
\begin{center}
\begin{tabular}{|c||c|c|c|c|c|c|c|c|} 
 \hline
 & $j_1$ & $j_2$ & $j_3$ & $j_4$ & $j_5$\\
 \hline
 $a$ & 1 & 1 & 1 & 1 & $k$ \\
 \hline
 $b$ & 1 & 1 & 1 & $k$ & 1 \\
 \hline
 $c$ & $k$ & $k$ & $k$ & $k$ & 1 \\
 \hline
\end{tabular}
\end{center}
$\s{MakeInitGroups}$ begins with some cost-minimizing allocation $\x$, say:
\begin{itemize}
    \item $\x_a = \{j_1, j_2, j_3, j_4\}$,
    \item $\x_b = \{j_5\}$,
    \item $\x_c = \{\}$.
\end{itemize}
Note that all chore payments are 1. Then, $a$ is the big earner, and she pEF1-envies an agent who is in her component, such as $b$. Indeed, $(b, j_1, a)$ is a shortest path from $b$ to $a$ in the MPB graph. We thus transfer $j_1$ from $a$ to $b$, giving us:
\begin{itemize}
    \item $\x_a = \{j_2, j_3, j_4\}$,
    \item $\x_b = \{j_5, j_1\}$,
    \item $\x_c = \{\}$.
\end{itemize}
It remains that $a$ is the big earner but she no longer pEF1-envies $b$, only $c$. As $c$ is in her component via shortest path $(c, j_5, b, j_2, a)$, we transfer $j_5$ to $c$, so we have:
\begin{itemize}
    \item $\x_a = \{j_2, j_3, j_4\}$,
    \item $\x_b = \{j_1\}$,
    \item $\x_c = \{j_5\}$.
\end{itemize}
Still $a$ is the big earner. While she pEF1-envies both $b$ and $c$, we note that $c$ is no longer in her component. We do however have shortest path ($b, j_2, a$). We transfer $j_2$ to $b$ to get:
\begin{itemize}
    \item $\x_a = \{j_3, j_4\}$,
    \item $\x_b = \{j_1, j_2\}$,
    \item $\x_c = \{j_5\}$.
\end{itemize}
Now we may take either $a$ or $b$ to be the big earner. Suppose we choose $b$. Only $a$ is in $b$'s component, and $b$ does not pEF1-envy her. We are thus ready to form agent group $N_1 = \{b, a\}$. Then, $c$ is left as the only agent not in a group, so she cannot pEF1-envy any other agent without a group. So, $c$ forms her own group $N_2 = \{c\}$. As all agents now belong to a group, $\s{MakeInitGroups}$ outputs the following:
\begin{itemize}
    \item $\x_a = \{j_3, j_4\}$,
    \item $\x_b = \{j_1, j_2\}$,
    \item $\x_c = \{j_5\}$.
    \item $N_1 = \{b, a\}$
    \item $N_2 = \{c\}$
\end{itemize}
It is easy to verify that this output satisfies properties (i), (ii), and (iii) of Section~\ref{sec:overviewsimple}.

\begin{example}\label{ex:Alg1extraK} Execution of Algorithm~\ref{alg:BalancedChores}, Algorithm~\ref{alg:efxpo}, Algorithm~\ref{alg:R1}, and Algorithm~\ref{alg:R1-1extra} 
\end{example}
Consider the chore allocation instance below where $k = 5$.
\begin{center}
\begin{tabular}{|c||c|c|c|c|c|c|c|c|} 
 \hline
 & $j_1$ & $j_2$ & $j_3$ & $j_4$ & $j_5$ & $j_6$ & $j_7$ & $j_8 \ldots j_{11}$\\
 \hline
 $a$ & 1 & 1 & 1 & 1 & 1 & 1 & $k$ & $k$ \\
 \hline
 $b$ & 1 & 1 & 1 & 1 & 1 & 1 & 1 & $k$ \\
 \hline
 $c$ & 1 & 1 & 1 & 1 & $k$ & 1 & 1 & $k$ \\
 \hline
\end{tabular}
\end{center}
Suppose Algorithm~\ref{alg:BalancedChores} outputs one agent group (so no transfers are required) with the following allocation:
\begin{itemize}
    \item $\x_a = \{j_1, j_2, j_3, j_8\}$,
    \item $\x_b = \{j_4, j_5, j_9, j_{10}\}$,
    \item $\x_c = \{j_6, j_7, j_{11}\}$.
\end{itemize}
Then, Algorithm~\ref{alg:efxpo} calls Algorithm~\ref{alg:R1}. Since $b$ EFX-envies both $a$ and $c$, we swap $j_9$ and $j_3$, giving us the following allocation where the only EFX-envy is from $a$ towards $b$:
\begin{itemize}
    \item $\x_a = \{j_1, j_2, j_9, j_8\}$,
    \item $\x_b = \{j_4, j_5, j_3, j_{10}\}$,
    \item $\x_c = \{j_6, j_7, j_{11}\}$.
\end{itemize}
Finally Algorithm~\ref{alg:R1-1extra} is called. We first execute on Line 14, where we swap $j_4$ and $j_7$, giving us:
\begin{itemize}
    \item $\x_a = \{j_1, j_2, j_9, j_8\}$,
    \item $\x_b = \{j_7, j_5, j_3, j_{10}\}$,
    \item $\x_c = \{j_6, j_4, j_{11}\}$.
\end{itemize}
We see that $a$ still EFX-envies $c$, so we execute Line 17, swapping $j_9$ with $j_4$. We now have:
\begin{itemize}
    \item $\x_a = \{j_1, j_2, j_4, j_8\}$,
    \item $\x_b = \{j_7, j_5, j_3, j_{10}\}$,
    \item $\x_c = \{j_6, j_9, j_{11}\}$.
\end{itemize}
Still we do not have EFX, as $c$ now EFX-envies $a$. Then, we execute on Line 23, transferring $j_6$ to $b$. This gives us our final EFX allocation:
\begin{itemize}
    \item $\x_a = \{j_1, j_2, j_4, j_8\}$,
    \item $\x_b = \{j_6, j_7, j_5, j_3, j_{10}\}$,
    \item $\x_c = \{j_9, j_{11}\}$.
\end{itemize}

\begin{example}\label{ex:Alg2extraK} Execution of Algorithm~\ref{alg:BalancedChores}, Algorithm~\ref{alg:efxpo}, Algorithm~\ref{alg:R1}, and Algorithm~\ref{alg:R1-2extra} 
\end{example}
Consider the chore allocation instance below where $k = 5$.
\begin{center}
\begin{tabular}{|c||c|c|c|c|c|c|} 
 \hline
 & $j_1$ & $j_2$ & $j_3$ & $j_4$ & $j_5$ & $j_6 \ldots j_{10}$\\
 \hline
 $a$ & 1 & 1 & 1 & $k$ & $k$ & $k$ \\
 \hline
 $b$ & 1 & 1 & 1 & 1 & 1 & $k$ \\
 \hline
 $c$ & 1 & 1 & 1 & 1 & 1 & $k$ \\
 \hline
\end{tabular}
\end{center}
Suppose Algorithm~\ref{alg:BalancedChores} outputs one agent group (so no transfers are required) with the following allocation:
\begin{itemize}
    \item $\x_a = \{j_1, j_2, j_6\}$,
    \item $\x_b = \{j_3, j_7, j_8\}$,
    \item $\x_c = \{j_4, j_5, j_9, j_{10}\}$.
\end{itemize}
Algorithm~\ref{alg:efxpo} again calls Algorithm~\ref{alg:R1}. Since both $b$ and $c$ EFX-envy $a$, and $|\x_c| > |\x_a|$, we transfer a $K$-chore from $c$ to $a$, giving us:
\begin{itemize}
    \item $\x_a = \{j_1, j_2, j_6, j_{10}\}$,
    \item $\x_b = \{j_3, j_7, j_8\}$,
    \item $\x_c = \{j_4, j_5, j_9\}$.
\end{itemize}
Now only $b$ EFX-envies $c$ and we move to Algorithm~\ref{alg:R1-2extra}. Here we execute at Line 37 and first transfer a 1-chore from $a$ to $c$, then transfer a 1-chore from $b$ to $c$. This gives us the EFX allocation:
\begin{itemize}
    \item $\x_a = \{j_2, j_6, j_{10}\}$,
    \item $\x_b = \{j_7, j_8\}$,
    \item $\x_c = \{j_4, j_5, j_9, j_1, j_3\}$.
\end{itemize}

\begin{example}\label{ex:Alg7} Execution of Algorithm~\ref{alg:BalancedChores}, Algorithm~\ref{alg:efxpo}, and Algorithm~\ref{alg:R2-1}
\end{example}
Consider the chore allocation instance below where $k = 5$.
\begin{center}
\begin{tabular}{|c||c|c|c|c|c|c|c|c|c|} 
 \hline
 & $j_1$ & $j_2$ & $j_3$ & $j_4$ & $j_5$ & $j_6$ & $j_7$ & $j_8$ & $j_9 \ldots j_{11}$\\
 \hline
 $a$ & 1 & 1 & 1 & 1 & 1 & 1 & $k$ & 1 & $k$ \\
 \hline
 $b$ & $k$ & $k$ & $k$ & $k$ & $k$ & $k$ & 1 & $k$ & $k$ \\
 \hline
 $c$ & $k$ & $k$ & $k$ & $k$ & $k$ & $k$ & 1 & 1 & $k$ \\
 \hline
\end{tabular}
\end{center}
Here Algorithm~\ref{alg:BalancedChores} creates two agent groups $\{a\}$ and $\{b, c\}$ with the following initial allocation:
\begin{itemize}
    \item $\x_a = \{j_1, j_2, j_3, j_4, j_5, j_6\}$,
    \item $\x_b = \{j_7, j_9\}$,
    \item $\x_c = \{j_8, j_{10}, j_{11}\}$.
\end{itemize}
Then, to balance the chores, we raise the payments of chores belonging to $a$ and transfer a chore to both $b$ and $c$, giving us:
\begin{itemize}
    \item $\x_a = \{j_1, j_2, j_3, j_4\}$,
    \item $\x_b = \{j_7, j_9, j_5\}$,
    \item $\x_c = \{j_8, j_{10}, j_{11}, j_6\}$.
\end{itemize}
Now the chores are balanced and Algorithm~\ref{alg:R2-1} is called. Since $c$ still EFX-envies $b$, we next execute Line 14, raising payments of $\x_a$ and then transferring a chore, say $j_4$, to $b$. This gives us the EFX allocation:
\begin{itemize}
    \item $\x_a = \{j_1, j_2, j_3\}$,
    \item $\x_b = \{j_7, j_9, j_5, j_4\}$,
    \item $\x_c = \{j_8, j_{10}, j_{11}, j_6\}$.
\end{itemize}

\begin{example}\label{ex:AlgR2-2} Execution of Algorithm~\ref{alg:BalancedChores}, Algorithm~\ref{alg:efxpo}, and Algorithm~\ref{alg:R2-2}
\end{example}
Consider the chore allocation instance below where $k = 5$.
\begin{center}
\begin{tabular}{|c||c|c|c|c|c|c|c|} 
 \hline
 & $j_1$ & $j_2$ & $j_3$ & $j_4$ & $j_5$ & $j_6$ & $j_7 \ldots j_{12}$\\
 \hline
 $a$ & 1 & 1 & 1 & 1 & 1 & $k$ & $k$ \\
 \hline
 $b$ & 1 & 1 & 1 & 1 & 1 & $k$ & $k$ \\
 \hline
 $c$ & $k$ & $k$ & $k$ & $k$ & $k$ & 1 & $k$ \\
 \hline
\end{tabular}
\end{center}
Here Algorithm~\ref{alg:BalancedChores} creates two agent groups $\{a, b\}$ and $\{c\}$ with the following initial allocation:
\begin{itemize}
    \item $\x_a = \{j_1, j_2, j_3, j_7\}$,
    \item $\x_b = \{j_4, j_5, j_8, j_9\}$,
    \item $\x_c = \{j_6, j_{10}, j_{11}, j_{12}\}$.
\end{itemize}
Since the chores are balanced, no transfers are needed. Algorithm~\ref{alg:efxpo} then calls Algorithm~\ref{alg:R2-2}. We satisfy the conditions in Lines 11, 14 and 18 and execute accordingly. We transfer a $K$-chore from $c$ to $a$ and a 1-chore from $a$ to $b$. This gives us the final EFX allocation:
\begin{itemize}
    \item $\x_a = \{j_2, j_3, j_7, j_{12}\}$,
    \item $\x_b = \{j_4, j_5, j_8, j_9, j_1\}$,
    \item $\x_c = \{j_6, j_{10}, j_{11}\}$.
\end{itemize}
\end{document}